\documentclass[10pt]{article}

\usepackage{algorithm}
\usepackage[noend]{algpseudocode}
\usepackage{amsmath}
\usepackage{amssymb}
\usepackage{amsthm}
\usepackage[english]{babel}
\usepackage{complexity}
\usepackage[T1]{fontenc}
\usepackage[a4paper,top=2.5cm,bottom=2.5cm,left=2.5cm,right=2.5cm]{geometry}
\usepackage[utf8]{inputenc}
\usepackage{mathtools}
\usepackage{multidef}
\usepackage{tikz}
\usetikzlibrary{arrows,backgrounds,calc,decorations.pathmorphing,decorations.markings,snakes,shapes}
\usepackage{tkz-graph}
\usepackage{url}
\usepackage{xspace}
\usepackage{xcolor}
\usepackage{fossacs}

\newif\ifqed
\qedtrue
\arxivtrue
\newcounter{lemma}
\let\qqed\qed
\renewcommand{\qed}{\ifqed \qedfalse \else \qedtrue \qqed \fi}
\newcounter{mycount}
\newtheorem{theorem}[mycount]{Theorem}
\newtheorem{lemma}[mycount]{Lemma}

\newtheorem{proposition}[mycount]{Proposition}
\theoremstyle{definition}
\newtheorem{definition}[mycount]{Definition}
\theoremstyle{remark}
\newtheorem{remark}[mycount]{Remark}

\begin{document}

\title{Dynamic Complexity of the Dyck Reachability}

\author{Patricia Bouyer and Vincent Jug\'e \\
LSV, CNRS \& ENS Cachan, Univ. Paris-Saclay, France \\
{\small This work is supported by EU under ERC EQualIS (FP7-308087).}}

\maketitle

\thispagestyle{plain}

\begin{abstract}
Dynamic complexity is concerned with updating the output of a problem
when the input is slightly changed. We study the dynamic complexity of
Dyck reachability problems in directed and undirected graphs, where updates
may add or delete edges.
We show a strong dichotomy between such problems, based on the size
of the Dyck alphabet. Some of them are \PTIME-complete
(under a strong notion of reduction) while the others
lie either in \DYN\FO or in \NL.
\end{abstract}

\section{Introduction}

\paragraph{Dynamic problems and dynamic complexity.} 

In this paper, we focus on the dynamic complexity of
some reachability problems. Standard complexity theory
aims at developing algorithms that, given an input of some problem,
compute an output as efficiently as possible.
Its dynamic variant is focused on algorithms that are capable of
efficiently updating the output after a small change of the
input~\cite{PI97,HI02,WS07}. 
Such algorithms may rely on auxiliary data about the
current instance of the problem, and update it
when the instance is modified.

A well-studied dynamic complexity class is \DYN\FO.
An algorithm is in \DYN\FO if the output and the auxiliary data
can be updated by \FO formulas after a small change of the input.
Variants of \DYN\FO include the class \DYN\FOP,
which allows polynomial-time precomputations,
and \DYN\TCZ, in which updates of the auxiliary data are performed
by \TCZ~circuits.

Consider the problem of reachability in directed graphs,
and update operations that consist in inserting or deleting
edges (one at a time). It was recently proven that this problem
belongs to the class \DYN\FO~\cite{DKMSZ15}, which had been conjectured
for decades.

Furthermore, like static complexity classes, dynamic complexity classes come with
natural notions of reduction. The class \DYN\FO is closed under bounded expansion first-order
reductions (hereafter called \bfo reductions), which are specific
\LOGSPACE reductions (\LOGSPACE is for logarithmic space).
A \bfo reduction from a problem to another one is a first-order mapping from instances
of the first problem to instances of the latter one, such that performing an update
operation on the instance of the first problem amounts to performing a bounded number of
update operations on the instance of the latter problem.
Similarly, the class \DYN\FOP is closed under bounded expansion first-order
reductions with polynomial-time precomputation (hereafter called \bfop reductions).

\paragraph{Reachability problems and language theory.}

Dyck reachability problems lie at the interface between two areas.
On the one hand, language theory is concerned with
handling descriptions of languages,
that is sets of words, with respect to various questions:
Is a language empty, finite or infinite?
What about the intersection or the union of two language?
Does a language contain a given word?
Among the best known and most simple classes of languages
are regular and context-free languages.
On the other hand, reachability problems deal with
the existence of paths in graphs,
and include questions such as:
Does there exist a path between two given vertices?
How long must be such paths?

Dyck reachability problems are focused on the existence
of paths in labeled graphs, whose labels belong to a given \emph{Dyck} language.
Dyck languages are languages of well-parenthesized words and, roughly speaking,
are the most simple context-free languages that are not regular.
The Dyck reachability problem in labeled directed acyclic graphs
was proven to be in \DYN\FO~\cite{WS07}, when considering two types
of update operations on labeled graphs, which are insertion and
deletion of edges. Whether this result extends to all labeled
directed graphs was then an open question.

\paragraph{Our contributions.}

We study this open question, and
we distinguish the Dyck reachability problem in two different ways.
Is the labeled graph directed or undirected?
How many symbols does the Dyck alphabet contain?

We prove that there exists a strong dichotomy between the dynamic
complexity of such problems, based on the size of the Dyck alphabet.
In the case of a unary Dyck alphabet, the Dyck reachability problem
lies in \NL~(non-deterministic logarithmic space), and even lies
in \DYN\FO~in the case of undirected graphs; this contrasts with the
case of binary Dyck alphabets, where we prove that the Dyck
reachability problem is \PTIME-complete under \bfop 
reductions. Furthermore, it is widely believed~\cite{PI97} that no
\PTIME-complete problems under \bfop reductions lie
in classes such as \DYN\FO~or the slightly broader class \DYN\FOP.

\paragraph{Related works.}

From its very inception 20 years ago, dynamic complexity has been a
framework of study for several variants of reachability problems
and language theory problems.
The class \DYN\FO was shown to contain reachability problems
in directed acyclic graphs~\cite{Dong95}, undirected graphs~\cite{PI97} and,
most recently, in all directed graphs~\cite{DKMSZ15};
regular and Dyck languages~\cite{PI97}, then all context-free languages~\cite{GMS12};
Dyck reachability in directed acyclic graphs~\cite{WS07}.

At the same time, finding natural problems that are
\NL- or \PTIME-complete (under \LOGSPACE~reductions) and
belong to low dynamic complexity classes such as
\DYN\FO, \DYN\FOP~or \DYN\TCZ~is an ongoing challenge.
All known \PTIME-complete problems lying in \DYN\FO~rely
on highly redundant inputs, hence may be seen as artificial~\cite{PI97}.
Hence, a notion of \emph{non-redundant projection}~\cite{MSVT94}
was introduced. Non-redundant projections are a special kind
of \PTIME reductions, which contains, in particular, \bfo and
\bfop reductions.

Hence, for every static complexity class $\calC$,
we define \emph{non-redundant} $\calC$-complete problems as
those problems that are $\calC$-complete both under
\LOGSPACE~reductions and under non-redundant projections.
Most canonical \PTIME-complete problems are non-redundant,
hence non-redundancy may be seen as a prerequisite for being a
``natural'' problem.

A breakthrough was the proof that the Dyck reachability problem
in acyclic directed graphs, which is a non-redundant \LOGCFL-complete problem,
belongs to \DYN\FO~\cite{WS07}.
It was then proved in~\cite{MVZ16} that
the reachability problem in labeled acyclic graphs, where path labels
are constrained to belong to a given context-free grammar
(and not only to a Dyck language), is in \Dyn\FO.
We prove here that the results of~\cite{MVZ16}
are unlikely to extend to all labeled graphs, or even to undirected graphs,
even in the simple case of two-letter Dyck languages.
This also allows us to answer negatively a question of Weber and Schwentick,
who asked in~\cite{WS07} whether ``the Dyck reachability problem might
be a non-redundant \PTIME-complete problem that allows efficient updates.''
\ifarxiv\else

\medskip

Complete proofs can be found in the research report~\cite{BG16}.
\fi

\section{Definitions}

\subsection{Dyck reachability problems}

A labeled directed graph is a triple $G = (V,L,E)$ where $V$ is a finite set of vertices,
$L$ is a finite set of labels and $E \subseteq V \times L \times V$ is a finite set of edges.
The graph $G$ is said to be \emph{unlabeled} if $L$ is a singleton set;
in that case, we may directly represent $G$ as a pair $(V,E)$ where $E \subseteq V \times V$.
The graph $G$ is also said to be \emph{undirected} if the relation $E$ is symmetric, i.e. if,
for every edge $(v,\theta,w)$ in $E$, the triple $(w,\theta,v)$ also belongs to $E$.

A path in the graph $G$ is a finite sequence of edges
$\pi = (v_1,\theta_1,w_1) \cdot (v_2,\theta_2,w_2) \cdot \ldots \cdot (v_k,\theta_k,w_k)$
such that $v_{i+1} = w_i$ for all $i \in \{1,\ldots,k-1\}$.
The vertex $v_1$ is called the \emph{source} of $\pi$, and $w_k$ is
called the \emph{sink} of $\pi$.
We also denote by $\lambda(\pi)$ the word $\theta_1 \cdot \ldots \cdot \theta_k$, which
is called the \emph{label} of $\pi$.

Assume that the label set $L$ is of the form
$L = \{\ell_1,\ldots,\ell_n\} \uplus \{\oell_1,\ldots,\oell_n\}$ for some integer $n \geq 1$.
The \emph{Dyck language} associated with $L$ is the context-free language
$\bfD_n$ built over the grammar: $S \to \varepsilon \mid \ell_1 \cdot S \cdot \oell_1 \cdot S
\mid \ldots \mid \ell_n \cdot S \cdot \oell_n \cdot S$,
where $\varepsilon$ is the empty word.
The set $\{\ell_1,\ldots,\ell_n\}$ is said to be the \emph{Dyck alphabet} of that language.

The \emph{$n$-letter Dyck reachability problem} asks whether, given two vertices $s$ and $t$ of $G$,
there exists a path in $G$, with source $s$ and sink $t$, and whose label belongs to the Dyck
language $\bfD_n$ (the actual value of the label set $L$ does not matter, as long as
its elements can be partitioned in $n$ ordered pairs).
The $n$-letter \emph{undirected} Dyck reachability problem is the restriction of that problem
to the case where the underlying graph $G$ is constrained to be undirected.

\subsection{Dynamic complexity}

In this paper, we study the dynamic complexity of Dyck reachability problems.
To that end, we first introduce briefly the formalisms of descriptive and dynamic
complexity here, and refer to~\cite{PI97,Imm99,Hes03} for more details.

Descriptive complexity aims at characterizing positive instances of a
problem using logical formulas. The input is then described as a
logical structure described by a set of $k$-ary predicates
(the~\emph{vocabulary}) over its universe. For example, a graph can be
described by a ternary predicate representing its edges, with the set
of states and of labels (usually identified with $\{1,\ldots,n\}$ for
some $n$) as the universe.  The~problem of deciding whether some state
has at most one outgoing edge can be described by the first-order
formula $\exists x. \forall y. \forall z. (E(x,y) \wedge E(x,z))
\Rightarrow (y=z)$. The~class \FO contains all problems that can be
characterized by such first-order formulas. This class corresponds to
the circuit-complexity class~\ACZ (under adequate
reductions)~\cite{BIS90}.

Dynamic complexity aims at developing algorithms that can efficiently update
the output of a problem when the input is slightly changed,
for example reachability of one vertex from another one in a graph.
We would like our algorithm to take advantage of
previous computations in order to very quickly decide the existence of a path
in the modified graph. 

Formally, a decision problem~$\sfS$ is a subset of the set of
$\tau$-structures $\mathsf{Struct}(\tau)$ built on a vocabulary~$\tau$. In~order to
turn~$\sfS$ into a dynamic problem~$\DYN\sfS$, we~need to define a
finite set of allowed updates.
For instance, we might use a $2$-ary operator $\textsf{ins}(x,y)$ that would insert an edge
between nodes~$x$ and~$y$. For a universe of size~$n$, the
set of update operations forms a finite alphabet, denoted by~$\Sigma_n$.
A~finite word in $\Sigma_n^\ast$ then
corresponds to a structure obtained by applying a 
sequence of update operations of~$\Sigma_n$ to the empty
structure $\calI_n$ over the vocabulary $\tau$.
The~language~$\DYN\sfS_n$ is defined as the set of those words in
$\Sigma_n^\ast$ that correspond to~structures of~$\sfS$, and
$\DYN\sfS$ is the union (over all~$n$) of all such languages.

A~dynamic machine is a uniform family $(M_n)_{n \in \bbN}$ of deterministic finite
automata $M_n= \langle Q_n,\Sigma_n,\delta_n,s_n,F_n \rangle$
over an update alphabet~$\Sigma_n$, with an update transition function $\delta_n$.
The set of states can be encoded as a
structure over some vocabulary~$\tau^{\mathrm{aux}}$, which contains the vocabulary $\tau$,
and corresponds
to a polynomial-size auxiliary data structure. Such a machine solves a
dynamic problem if $\DYN\sfS_n=\calL(M_n)$ for all~$n$. It~is in the
dynamic complexity class~$\calC'\text{-}\DYN\calC$ (or simply~$\DYN\calC$
if~$\calC=\calC'$) if the update transition function and accepting set
can be computed in~$\calC$, while the initial state
can be computed in~$\calC'$. In other
words, solving the initial instance of the problem can be done in $\calC'$,
and after any update of
the input (specified by some letter of $\Sigma_n$), further
calculations to solve the problem on that new instance are restricted
to the class $\calC$. Of course, for the a dynamic complexity
class~$\calC'\text{-}\DYN\calC$ to have some interest, the class $\calC$
should be easier than the static complexity class of the original
problem.

In~this paper, we only consider the case where $\calC=\FO$, and where
$\calC' = \FO$ or $\calC' = \PTIME$, meaning that first-order formulas
will be used to describe how predicates are updated along transitions,
and that we may make use of polynomial-time precomputations.  As a
convention, we will denote the class $\PTIME\text{-}\DYN\FO$ by
$\DYN\FOP$, and we recall that the simple notation $\DYN\FO$ is
for $\FO\text{-}\DYN\FO$.

\subsection{Dynamic reductions}

Dynamic complexity comes with the notion of dynamic
reductions~\cite{PI97}.  Let $\calC$ be a complexity class. A
(static) $\calC$ reduction from a decision problem
$\calP$ to another decision problem $\calQ$ is a mapping in $\calC$
from the instances of $\calP$ to the instances of $\calQ$ that
associates every positive instance of $\calP$ with a positive instance
of $\calQ$, and every negative instance of $\calP$ with a negative
instance of $\calQ$. Standard \PTIME-completeness results
use \LOGSPACE~reductions~\cite{GHR95}.

A \emph{dynamic reduction} from a dynamic problem $\calP$ (with
vocabulary $\tau_1$) to another dynamic problem $\calQ$ (with
vocabulary $\tau_2$) is a mapping from $\mathsf{Struct}(\tau_1)$ to
$\mathsf{Struct}(\tau_2)$ such that:
\begin{itemize}
  \item every positive (respectively, negative) instance of $\calP$
  is mapped to a positive (respectively, negative) instance of $\calQ$;
\item every update on an instance $i_1$ of $\calP$ results in a
  \emph{well-behaved} sequence of updates on the instance $i_2$ of
  $\calQ$ to which $i_1$ is mapped.
\end{itemize}
Dynamic reductions have therefore several parameters: the
complexity class to which the mapping belongs, and the sequences of
updates that are allowed.

The dynamic classes $\DYN\FO$ and $\DYN\FOP$ are respectively closed
under bounded expansion first-order (bfo for short) and bounded expansion
first-order with polynomial-time precomputation
(bfo\textsuperscript{+} for short) reductions~\cite{PI97}.  A dynamic
reduction $\mu$ from $\calP$ to $\calQ$ is bfo\textsuperscript{+} if
it is a \FO~reduction and if every update on an instance $i_1$ of
$\calP$ results in a bounded sequence of \FO~updates on its
image $\mu(i_1)$.  If, furthermore, the empty structure $\calI_1$ is
mapped to a structure $\mu(\calI_1)$ that can be obtained by applying
a bounded sequence of \FO~updates on the empty structure
$\calI_2$, then we say that $\mu$ is bfo.

Note that dynamic reductions can be applied to the class \PTIME~(which
coincides with the class \DYN\PTIME, under the assumption that
updates are one-bit input changes).
So, being \PTIME-hard for bfo\textsuperscript{+}
reductions is a priori stronger than being \PTIME-hard for \LOGSPACE~reductions. 
Furthermore, it is known that the classes of bfo and of
bfo\textsuperscript{+} reductions are closed under composition
and that the circuit value problem is a \PTIME-complete problem for bfo\textsuperscript{+}
reductions~\cite{PI97}.
Hence, every \PTIME~problem to which the circuit value problem is
bfo\textsuperscript{+}-reducible is also \PTIME-complete problem for bfo\textsuperscript{+}
reductions.

\subsection{Main result}

We are now in a position to formally present our main result.

\begin{theorem}
\label{thm:main}
The $1$-letter Dyck reachability problem is in $\NL$, and the
$1$-letter undirected Dyck reachability problem is in $\NL \cap
\DynFO$.  Furthermore, for all integers $n \geq 2$, both the
$n$-letter Dyck reachability problem and the $n$-letter undirected
Dyck reachability problems are $\PTIME$-complete for
bfo\textsuperscript{+} reductions.
\end{theorem}

\begin{remark}
  Note that $\NL \cap \DynFO$ is not known to be strictly included in
  $\NL$. Nevertheless, the case of undirected graph appears to be
  ``easier'' than the case of directed graphs in the $1$-letter
  case. Hence, the $\PTIME$-hardness of both cases for alphabets with
  at least two letters appears rather unexpected.
\end{remark}

\section{One-letter (undirected) Dyck reachability problems}

We prove here the first part of Theorem~\ref{thm:main}, that is we
assume $n=1$.
We first observe that the $1$-letter Dyck reachability problem is equivalent to
a standard reachability problem in one-counter automata (without zero-tests),
which is known to belong to \NL~\cite{DG09,HKOW09}.
The $1$-letter undirected Dyck reachability problem is a restriction of the
$1$-letter undirected Dyck reachability problem, hence it is in $\NL$ as well.
Furthermore, we make the following claim.

\begin{proposition}\label{pro:thm-2-n=1}
Let $s$ and $t$ be two distinct vertices of an undirected labeled graph
$G = (V,E,L)$, with $L = \{\ell_1,\oell_1\}$. There exists a Dyck path from $s$ to $t$ in $G$
if and only if:
\begin{itemize}
 \item the set $\{x \in V \mid (s,\ell_1,x) \in E\}$ is non-empty;
 \item the set $\{y \in V \mid (t,\oell_1,y) \in E\}$ is non-empty;
 \item there exists a path of even length from $s$ to $t$ in $G$.
\end{itemize}
\end{proposition}

\begin{proof}
First, if there exists a Dyck path $\pi = (v_i,\lambda_i,v_{i+1})_{0 \leq i < k}$ with
$s = v_0$ and $t = v_k$, then $\lambda_0 = \ell_1$, $\lambda_{k-1} = \oell_1$, and
the sets $\{0 \le i < k \mid \lambda_i = \ell_1\}$ and $\{0 \le i < k \mid \lambda_i = \oell_1\}$ have the same
cardinality, which proves that $k$ is an even number.

Conversely, assume that $s \neq t$ and that the three conditions of Proposition~\ref{pro:thm-2-n=1} hold.
Let $\pi = (v_i,\lambda_i,v_{i+1})_{0 \leq i < 2k}$ be a path of length $2k$ from $s$ to $t$ in $G$,
for some integer $k \geq 1$.
Let $\kappa$ be the cardinality of the set $\{0 \leq i < 2k \mid \lambda_i = \ell_1\}$ and
let $\ol{\kappa}$ be the cardinality of the set $\{0 \leq i < 2k \mid \lambda_i = \oell_1\}$.
Since $\kappa + \ol{\kappa} = 2k$, we have $\ol{\kappa} - \kappa = 2(k-\kappa)$.

Furthermore, consider vertices $x, y \in V$ such that $(s,\ell_1,x)$
and $(t,\oell_1,y)$ belong to $E$. Since the graph is undirected,
there exist also edges $(x,\ell_1,s)$ and $(y,\oell_1,t)$.
Let $\rho_1$ be the length-$2$ circuit $(s,\ell_1,x) \cdot
(x,\ell_1,s)$, and let $\rho_2$ be the length-$2$ circuit
$(t,\oell_1,y) \cdot (y,\oell_1,x)$.
One checks easily that the path $\rho_1^k \cdot \pi \cdot \rho_2^\kappa$ is a Dyck path in $G$,
where $\rho_1^k$ is the concatenation of $k$ occurrences of $\rho_1$, and
$\rho_2^\kappa$ is the concatenation of $\kappa$ occurrences of
$\rho_2$. 
\end{proof}

Hence, checking whether there exists a Dyck path from $s$ to $t$ in
$G$ amounts to checking whether $s = t$ or, if $s \neq t$, whether the
sets $\{x \in V \mid (s,\ell_1,x) \in E\}$ and $\{y \in V \mid
(t,\oell_1,y) \in E\}$ are non-empty, and whether there exists a path
of even length from $s$ to $t$ in $G$.  The first statements can be
checked directly in \FO, and the latter one can be checked in
\Dyn\FO~\cite{PI97,DKMSZ15}.  This completes the proof of the first part of
Theorem~\ref{thm:main} in the case $n = 1$.

\begin{remark}
  Note that this proof heavily relies on the property that the graph
  be undirected.
  In fact, the $1$-letter Dyck reachability problem (over directed
  graphs) is bfo\textsuperscript{+}-reducible to the problem of
  computing distances in directed graphs, whose membership in \DYN\FO
  or \DYN\FOP is a long-standing open question~\cite{Hes03a,DKMSZ15}.
  The reduction is as follows.

  Given an unlabeled directed graph $G = (V,E)$, equip each edge with
  a label $\ell_1$, and add self-loops (with the label $\ell_1$)
  around each vertex in $V$.  Then, for all vertices $v \in V$, add
  $n$ vertices $(v,1), \ldots, (v,n)$, where $n = |V|$, and add edges
  with the label $\oell_1$ from $v$ to $(v,1)$ and from $(v,i)$ to
  $(v,i+1)$, for all $i$.  It comes at once that the distance (in the
  original graph $G$) from a vertex $s$ to a vertex $t$ is $k$ if and
  only if there exists a Dyck path (in the extended, labeled graph)
  from $s$ to $(t,k)$ but not to $(t,k-1)$.

  Furthermore, the proof of~\cite{Hes03a} showing that distances in graphs can
  be computed in \DYN\TCZ~does not extend to the
  $1$-letter Dyck reachability, whose precise dynamic complexity
  remains therefore unknown.
\end{remark}

\section{Two-letter Dyck reachability problem}

We prove now that, for all integers $n \geq 2$, the $n$-letter Dyck
reachability problem is \PTIME-complete for bfo\textsuperscript{+}
reductions.  

We first introduce two auxiliary problems.

\begin{enumerate}
\item Let $G = (V,E)$ be an unlabeled directed graph, let
  $(V_\wedge,V_\vee)$ be a partition of $V$, and let $s$ and $t$ be
  two marked vertices of $G$. The
  \emph{alternating reachability problem} asks whether $s$ belongs to
  the smallest subset $X$ of $V$ such that all of $\{t\}$, $\{x \in
  V_\vee \mid \exists y \in X \text{ s.t. } (x,y) \in E\}$ and $\{x
  \in V_\wedge \mid \forall y \in V, (x,y) \in E \Rightarrow y \in
  X\}$ are subsets of $X$. \label{def-ARP}

  Note that this problem could be alternatively and equivalently
  defined using the notion of winning state in a two-player
  turn-based zero-sum reachability game. However we choose the above
  definition using a fixed point to avoid defining the notion of
  winning strategies.

\item Let $G=(V,E,L)$ be a labeled directed graph with set of labels
  $L = V \cup \{\overline{v} \mid v \in V\} \cup \{\bullet\}$, where
  $\bullet$ is a fresh label symbol, and let $s$ and $t$ be two marked
  vertices of $G$.  A \emph{near-Dyck word} is an element of the set
  $\bfD'$ built over the grammar: $S \to \varepsilon \mid S \cdot
  \bullet \cdot S \mid v \cdot S \cdot \overline{v} \cdot S$ (for all
  $v \in V$).  The \emph{near-Dyck reachability problem} asks whether
  there exists a path $\pi$ in $G$, with source $s$, sink $t$, and
  whose label belongs to $\mathbf{D}'$.
\end{enumerate}

While it is well-known that the alternating
reachability problem
is \PTIME-hard for standard logarithmic-space~reductions, it is also the
case that it is \PTIME-hard for bfo\textsuperscript{+}
reductions~\cite{PI97}.  Hence, we show in the two next
subsections that there exists a bfo\textsuperscript{+} reduction from
the alternating reachability problem to the near-Dyck reachability
problem, and that there exists a bfo\textsuperscript{+} reduction from
that latter problem to the $2$-letter Dyck reachability problem.  It
will follow that the $2$-letter (and, therefore, the $n$-letter) Dyck
reachability problem is \PTIME-hard for bfo\textsuperscript{+}
reductions.

\begin{remark}
  By carefully adaptating the reductions below, we might
  transform them into bfo reductions only. However, since the
  the alternating reachability problem is only known to be
  \PTIME-complete under bfo\textsuperscript{+} reductions,
  such a transformation would not provide us with
  stronger \PTIME-completeness results for the $n$-letter
  Dyck reachability problem.
\end{remark}

On the other hand, the $n$-letter Dyck reachability problem is clearly in \PTIME,
as witnessed by the following algorithm, which will complete the proof of Theorem~\ref{thm:main}
in the case of the $n$-letter Dyck reachability problem.

\begin{algorithm}
\caption{Finding an $n$-letter Dyck path}\label{guess-dyck-path}
\begin{algorithmic}[1]
\State $S \gets \{(x,x) \mid x \in V\}$
\While{there exists pairs $(u,v) \notin S$ and $(u',v') \in S$ and edges $(u,\lambda,u')$ and $(v',\ol{\lambda},v)$
for some label $\lambda \in \{\ell_1,\ldots,\ell_n\}$}
  \State $S \gets S \cup \{(u,v)\}$
\EndWhile
\State \textbf{return} ($(s,t) \in S$)
\end{algorithmic}
\end{algorithm}

\subsection{From the near-Dyck reachability problem to the Dyck reachability problem}
\label{section:near-dyck-to-dyck}

Let $G=(V,E,L)$ be a labeled directed graph with set of labels $L = V
\cup \ol{V} \cup \{\bullet\}$ (where $\ol{V} = \{\overline{v} \mid v
\in V\}$), and let $s$ and $t$ be two marked vertices of $G$.

We fix the new alphabet $\calL =
\{a,b,\overline{a},\overline{b}\}$ and an arbitrary bijection
$\theta : i \mapsto v_i$ between the set $\{1,\ldots,n\}$ and $V$
(for $n = |V|$).  We will expand the graph so that the letter
$\bullet$ will be encoded using the word $a \cdot \ol{a}$, and each
letter $v_i \in V$ (respectively, $\ol{v_i} \in \ol{V}$) will be represented
by the word $a^i \cdot b \cdot a^{n+1-i}$ (respectively, $\ol{a}^{n+1-i}
\cdot \ol{b} \cdot \ol{a}^i$).

Formally, let $\calG = (\calV,\calE,\calL)$ be the labeled directed
graph defined by:
\begin{itemize}
\item $\calV = V \cup \bigl(V \times \{\bullet\}\bigr) \cup \bigl(V
  \times (V \cup \ol{V}) \times \{0,1,\ldots,n\}\bigr)$;
 \item $\calL = \{a,b,\overline{a},\overline{b}\}$;
 \item $\calE = \calE_1 \cup \calE_2$, where
\end{itemize}
\smallskip
\begin{equation*}
  \begin{split}
  \calE_1 =~& \{x \xrightarrow{a} (x,\bullet) \mid x \in V\} ~ \cup \\
  & \{x \xrightarrow{a} (x,v,0) \mid x,v \in V\} ~ \cup \\
  & \{(x,v,i) \xrightarrow{a} (x,v,i+1) \mid x,v \in V, 0 \leq i \leq n-1, v \neq v_{i+1}\} ~ \cup \\
  & \{(x,v,i) \xrightarrow{b} (x,v,i+1) \mid x,v \in V, 0 \leq i \leq n-1, v = v_{i+1}\} ~ \cup \\
  & \{x \xrightarrow{\overline{a}} (x,\ol{v},n) \mid x,v \in V\} ~ \cup \\
  & \{(x,\ol{v},i+1) \xrightarrow{\overline{a}} (x,\ol{v},i) \mid x,v \in V, 0 \leq i \leq n-1, v \neq v_{i+1}\} ~ \cup \\
  & \{(x,\ol{v},i+1) \xrightarrow{\overline{b}} (x,\ol{v},i) \mid x,v \in V, 0 \leq i \leq n-1, v = v_{i+1}\} \text{ and } \\
  \calE_2 =~& \{(x,\bullet) \xrightarrow{\overline{a}} y \mid x \xrightarrow{\bullet} y \in E\} ~ \cup \\
  & \{(x,v,n) \xrightarrow{a} y \mid x \xrightarrow{v} y \in E\} ~ \cup \\
  & \{(x,\ol{v},0) \xrightarrow{\overline{a}} y \mid x \xrightarrow{\overline{v}} y \in E\},
  \end{split}
\end{equation*}
and in which we mark the vertices $s$ and $t$.

Each sequence of transitions $x \xrightarrow{a} (x,v_{i+1},0) \xrightarrow{a} \ldots
\xrightarrow{a} (x,v_{i+1},i) \xrightarrow{b} (x,v_{i+1},i+1) \xrightarrow{a} \ldots
\xrightarrow{a} \dots (x,v_{i+1},n)$ prepares the encoding of some
edge leaving $x$ with label $v_{i+1}$. If there is some edge $x
\xrightarrow{v_{i+1}} y$ in the original graph, then only one edge
$(x,v_{i+1},n) \xrightarrow{a} y$ needs to be added: this is the role of the edges in
$\calE_2$. We use a similar encoding for edges labeled by $\ol{v_{i+1}}$,
and an even simpler encoding for edges labeled by $\bullet$.

\begin{proposition}
\label{pro:near-dyck-to-dyck}
There exists a near-Dyck path from $s$ to $t$ in $G$ if and only if
there exists a Dyck path from $s$ to $t$ in $\calG$.
\end{proposition}

\begin{proof}
First, for every pair $(u,v) \in \calV^2$, there exists at most one edge in $\calE$ with
source $u$ and sink $v$. Henceforth, we omit representing labels of edges and of paths in $\calG$.

We further define two mappings $\varphi$ and $\psi$.
The mapping $\varphi$ identifies every label $\lambda \in L$ with a word $\varphi(\lambda) \in \calL^\ast$, as follows:
\begin{equation*}
  \begin{split}
  \varphi(\bullet) =~& a \cdot \ol{a} \\
  \varphi(v_i) =~& a^i \cdot b \cdot a^{n+1-i} \text{ for all $i \in \{1,\ldots,n\}$} \\
  \varphi(\ol{v_i}) =~& \ol{a}^{n+1-i} \cdot \ol{b} \cdot \ol{a}^i \text{ for all $i \in \{1,\ldots,n\}$}
  \end{split}
\end{equation*}
and extends immediately to a morphism from $L^\ast$ to $\calL^\ast$ that
maps every near-Dyck word $w \in \bfD'$ to a Dyck word $\varphi(w) \in \bfD$.
The mapping $\psi$ identifies every edge $e \in E$ with a path $\psi(e)$ in $\calG$, as follows:
\begin{equation*}
  \begin{split}
  \psi(x \xrightarrow{\bullet} y) =~& (x \to (x,\bullet) \to y) \\
  \psi(x \xrightarrow{v} y) =~& (x \to (x,v,0) \to \ldots \to (x,v,n)
  \to y) \text{ for all $v \in V$} \\
  \psi(x \xrightarrow{\ol{v}} y) =~& (x \to (x,\ol{v},n) \to \ldots
  \to (x,\ol{v},0) \to y) \text{ for all $\ol{v} \in \ol{V}$}
  \end{split}
\end{equation*}
and extends immediately to a morphism that maps every path in $G$ to a
path in $\calG$.  The relation $\lambda(\psi(e)) =
\varphi(\lambda(e))$ holds for all edges $e \in E$, and therefore
extends to all paths $\pi$ in $G$.  Hence, a path $\pi$ in $G$ is
near-Dyck if and only if the path $\psi(\pi)$ in $\calG$ is Dyck.

In addition, let us call \emph{nominal} paths in $\calG$ the paths
that belong to the set $\{\psi(e) \mid e \in E\}$, and \emph{generic}
paths in $\calG$ the concatenations of nominal paths.  Nominal paths
are the minimal paths whose source and sink both belong to the subset
$V$ of $\calV$.  Hence, every path $\pi$ from $s$ to $t$ in $\calG$ is
generic, thus $\pi$ is the image by $\psi$ of some path
$\psi^{-1}(\pi)$ from $s$ to $t$ in $G$. 
\end{proof}

The graph $\calG$ is \FO-definable as a function of $G$ and of
the bijection $\theta : \{1,\ldots,n\} \mapsto V$,
and adding\slash deleting an edge in $E$ amounts to
adding\slash deleting exactly one edge in $\calE_2$.
Since $\theta$ can be precomputed in \PTIME,
and due to Proposition~\ref{pro:near-dyck-to-dyck},
the near-Dyck reachability problem is therefore bfo\textsuperscript{+}-reducible
to the Dyck reachability problem.

\subsection{From the alternating reachability problem to the near-Dyck
  reachability problem}
\label{section:alternating-to-near-dyck}

Let $G = (V,E)$ be an unlabeled directed graph, let
$(V_\wedge,V_\vee)$ be a partition of $V$, and let $s$ and $t$ be two marked
vertices of $G$.

We fix the new alphabet $\calL = V \cup \{\ol{v} \mid v \in V\} \cup \{\bullet\}$,
where $\bullet$ is a fresh label symbol, and an arbitrary bijection
$\theta : i \mapsto v_i$ between the set $\{1,\ldots,n\}$ and $V$
(for $n = |V|$). Recall the definition of the set $X$ as a
smallest fixed point page~\pageref{def-ARP}.
We will expand the graph $G$ so that the statement $s
\in X$ be equivalent to the existence of some Dyck path from $t$ to $s$ (in
the expanded graph). 
More precisely, for every vertex $x \in G$, we want the following statements to be
equivalent to each other:
\begin{enumerate}
 \item the vertex $x$ belongs to $X$;
 \item there exists a Dyck path from $t$ to $x$ (in the expanded graph);
 \item there exists a path from $t$ to itself whose label reduces to
   the letter $x$
   (in the expanded graph), where the reduction consists in
     deleting recursively the well-balanced subwords
     $\bullet$ and $v \cdot \ol{v}$ (for $v \in V$).
\end{enumerate}

The equivalence $2 \Leftrightarrow 3$ is ensured by drawing an
edge $x \xrightarrow{x} t$ for every vertex $x \in V$ and
never using the label $x$ for any other edge.
The equivalence $1 \Leftrightarrow 2$ is obtained by induction,
using the following construction (and the already proven equivalence
$2 \Leftrightarrow 3$).
For every vertex $x \in V_\vee$ and every edge $(x,y) \in E$,
we draw an edge $y \xrightarrow{\bullet} x$;
for every vertex $x \in V_\wedge$, we construct one unique path $\inn(x)$
from $t$ to $x$, and we build it so that its reduced label is
$\prod_{y \in \{z \mid (x,z) \in E\}} \ol{v_y}$.

\begin{figure}[ht]
\begin{center}
\begin{tikzpicture}[scale=0.68]
\SetGraphUnit{4}
\tikzset{EdgeStyle/.append style = {>=stealth}}

\Vertex[x=0,y=0,L={$\wedge$}]{1}
\Vertex[x=-2,y=-2,L={$\vee$}]{2}
\EA[L={$\wedge$}](2){3}
\Vertex[x=-2,y=-5,L={$\vee$}]{4}
\EA[L={$\vee$}](4){5}
\node[anchor=south] at (0,0.5) {$s=v_1$};
\node[anchor=south] at (-2,-1.5) {$v_2$};
\node[anchor=south] at (2,-1.5) {$v_3$};
\node[anchor=north] at (-2,-5.5) {$v_4$};
\node[anchor=north] at (2,-5.5) {$t=v_5$};

\node[anchor=south] at (0,1.5) {Graph $G$};

\Edge[style={->}](1)(2)
\Edge[style={->}](1)(3)
\Edge[style={bend right=15,->}](2)(3)
\Edge[style={bend right=15,->}](3)(2)
\Edge[style={->}](2)(4)
\Edge[style={->}](3)(5)
\Edge[style={->}](4)(5)
\Edge[style={bend left=15,->}](5)(2)

\Vertex[x=8,y=0,L={$v_1$}]{1}
\Vertex[x=6,y=-2,L={$v_2$}]{2}
\EA[L={$v_3$}](2){3}
\Vertex[x=6,y=-5,L={$v_4$}]{4}
\EA[L={$v_5$}](4){5}

\Edge[style={->}](3)(2)
\Edge[style={->}](4)(2)
\Edge[style={->}](5)(4)
\Edge[style={->}](2)(5)

\Edge[style={->}](1)(5)
\Edge[style={->}](3)(5)
\Edge[style={bend right=20,->}](4)(5)
\Loop[dist=1cm,dir=SO,style={thick,->}](5)

\node[anchor=south] at (8,-2.1) {$\bullet$};
\node[anchor=south] at (8,-5.1) {$\bullet$};
\node[anchor=east] at (6.1,-3.5) {$\bullet$};
\node[anchor=north east] at (9.1,-2.4) {$v_1$};
\node[anchor=north east] at (8.1,-3.4) {$\bullet,v_2$};
\node[anchor=west] at (9.9,-3.5) {$v_3$};
\node[anchor=north] at (8,-5.5) {$v_4$};
\node[anchor=north] at (10,-5.9) {$v_5$};

\SetGraphUnit{1}

\vs

\Vertex[x=12.5,y=-5,L={~}]{30}
\NO[L={~}](30){31}
\NO[L={~}](31){32}
\NO[L={~}](32){33}
\NO[L={~}](33){34}
\NO[L={~}](34){35}

\Vertex[x=14.5,y=-5,L={~}]{10}
\NO[L={~}](10){11}
\NO[L={~}](11){12}
\NO[L={~}](12){13}
\NO[L={~}](13){14}
\NO[L={~}](14){15}

\Edge[style={bend right=20,->}](5)(10)
\Edge[style={->}](10)(11)
\Edge[style={->}](11)(12)
\Edge[style={->}](12)(13)
\Edge[style={->}](13)(14)
\Edge[style={->}](14)(15)
\Edge[style={bend right=20,->}](15)(1)

\Edge[style={->}](5)(30)
\Edge[style={->}](30)(31)
\Edge[style={->}](31)(32)
\Edge[style={->}](32)(33)
\Edge[style={->}](33)(34)
\Edge[style={->}](34)(35)
\Edge[style={->}](35)(3)

\node[anchor=south] at (11.25,-5.1) {$\bullet$};
\node[anchor=west] at (12.45,-4.6) {$\bullet$};
\node[anchor=west] at (12.45,-3.6) {$\ol{v_2}$};
\node[anchor=west] at (12.45,-2.6) {$\bullet$};
\node[anchor=west] at (12.45,-1.6) {$\bullet$};
\node[anchor=west] at (12.45,-0.6) {$\ol{v_5}$};
\node[anchor=south] at (11.25,-1.05) {$\bullet$};

\node[anchor=north] at (11.25,-5.4) {$v_4$};
\node[anchor=west] at (14.45,-4.6) {$\bullet$};
\node[anchor=west] at (14.45,-3.6) {$\ol{v_2}$};
\node[anchor=west] at (14.45,-2.6) {$\ol{v_3}$};
\node[anchor=west] at (14.45,-1.6) {$\bullet$};
\node[anchor=west] at (14.45,-0.6) {$\bullet$};
\node[anchor=south] at (11.25,0.6) {$\bullet$};

\draw [decorate,decoration=brace] (12,1) -- (13,1) node
[midway,above] {\begin{tabular}{c} path \\ $\inn(v_3)$ \end{tabular}};
\draw [decorate,decoration=brace] (14,1) -- (15,1) node
[midway,above] {\begin{tabular}{c} path \\ $\inn(v_1)$ \end{tabular}};

\node[anchor=south] at (8,1.5) {Graph $\calG$};
\end{tikzpicture}
\end{center}
\caption{Graphs $G$ and $\calG$. A near-Dyck
  path witnessing that $s \in X$ is $v_5 \xrightarrow{v_5}
  v_5 \xrightarrow{\bullet} v_4 \xrightarrow{\bullet} v_2
  \xrightarrow{v_2} v_5 \cdot \inn(v_3) \cdot v_3 \xrightarrow{v_3}
  v_5 \xrightarrow{\bullet} v_4 \xrightarrow{\bullet} v_2
  \xrightarrow{v_2} v_5 \cdot \inn(v_1) \cdot v_1$}
\label{fig:and-or-to-near-dyck}
\end{figure}
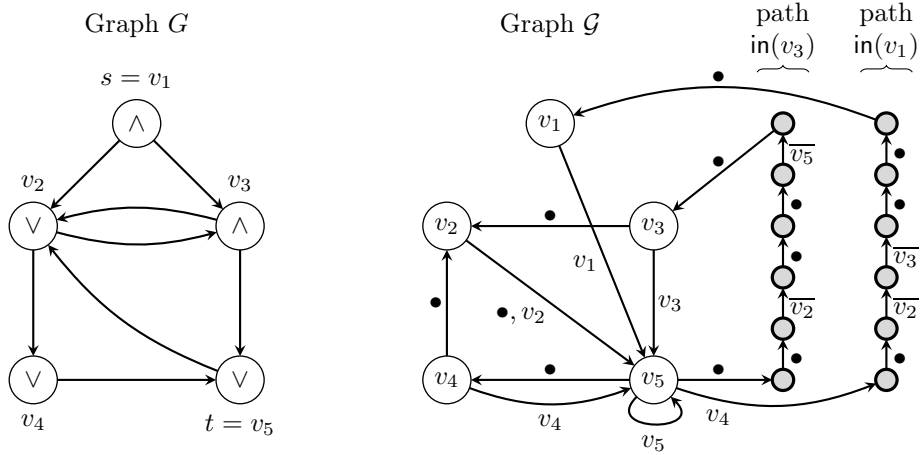

Formally, let $\calG = (\calV,\calE,\calL)$ be the labeled directed graph defined by:
\begin{itemize}
 \item $\calV = V \cup \bigl(V_\wedge \times \{0,1,\ldots,n\}\bigr)$;
 \item $\calL = V \cup \{\overline{v} \mid v \in V\} \cup \{\bullet\}$;
 \item $\calE = \calE_1 \cup \calE_2$, where
\end{itemize}
\smallskip
\begin{equation*}
  \begin{split}
  \calE_1 = ~& \{x \xrightarrow{x} t \mid x \in V\} ~ \cup ~
  \{t \xrightarrow{\bullet} (x,0) \mid x \in V_\wedge\} ~ \cup ~
  \{(x,n) \xrightarrow{\bullet} x \mid x \in V_\wedge\} \text{ and } \\
  \calE_2 = ~& \{x \xrightarrow{\bullet} y \mid y \in V_\vee, (y,x) \in E\} ~ \cup \\
  & \{(x,i) \xrightarrow{\overline{v_{i+1}}} (x,i+1) \mid x \in V_\wedge, 0 \leq i \leq n-1, (x,v_{i+1}) \in E\} ~ \cup \\
  & \{(x,i) \xrightarrow{\bullet} (x,i+1) \mid x \in V_\wedge, 0 \leq i \leq n-1, (x,v_{i+1}) \notin E\},
  \end{split}
\end{equation*}
and in which we mark the vertices $t$ and $s$.

The construction is illustrated in Fig.~\ref{fig:and-or-to-near-dyck}.

\begin{proposition}
\label{pro:alternating-to-near-dyck}
Let $X$ be the smallest subset of $V$ such that all of $\{t\}$, 
$\{x \in V_\vee \mid \exists y \in X
\text{ s.t. } (x,y) \in E\}$ and $\{x \in V_\wedge \mid
\forall y \in V, (x,y) \in E \Rightarrow y \in X\}$ are subsets of $X$.
The vertex $s$ belongs to $X$ if and only if there exists a near-Dyck path from $t$ to $s$ in $\calG$.
\end{proposition}

\begin{proof}
  Like in Section~\ref{section:near-dyck-to-dyck}, observe that, for
  every pair $(u,v) \in \calV^2$, there exists at most one edge in
  $\calE$ with source $u$ and sink $v$. Henceforth, we will sometimes
  omit representing labels of edges and of paths in $\calG$.
  Conversely, for all $v \in V$, the edge $v \xrightarrow{v} t$ is
  the only edge in $\calE$ with label $v$.

  Furthermore, for all $x \in V_\wedge$, we denote by $\inn(x)$
  the path $t \to (x,0) \to (x,1) \to \ldots \to (x,n) \to x$.  Every
  non-empty path with source in $V$ and sink $x$ must end with the
  sub-path $\inn(x)$.

Observe that $X$ is inductively defined as $X = \bigcup_{i \geq 0} X_i$,
where $X_0 = \{t\}$ and
\begin{equation*}
  \begin{split}
  X_{i+1} = ~& X_i \cup \{z \in V_\vee \mid \exists y \in X_i
    \text{ s.t. } (z,y) \in E\} ~\cup \\
  & \{z \in V_\wedge \mid
    \forall y \in V, (z,y) \in E \Rightarrow y \in X_i\}.
  \end{split}
\end{equation*}

Then, we say that a sequence $w_0,\ldots,w_k$ of vertices of $G$ is
\emph{well-ordered} if $t = w_0$ and, for all $i \in \{1,\dots,k\}$:
\begin{itemize}
\item if $w_i \in V_\vee$, then $\{z \in V \mid (w_i,z) \in E\} \cap
  \{w_0,\ldots,w_{i-1}\} \neq \emptyset$;
\item if $w_i \in V_\wedge$, then $\{z \in V \mid (w_i,z) \in E\}
  \subseteq \{w_0,\ldots,w_{i-1}\}$.
\end{itemize}

For every $x \in X$, let $\iota(x)$ be the smallest $i$ such that $x
\in X_i$. Define any linear ordering $\prec$ over $X$ such that
$\iota(x) < \iota(y)$ implies $x \prec y$ (in particular the order
between two vertices belonging to $X_i \setminus X_{i-1}$ is
arbitrary).
If $w_0 \prec w_1 \prec \ldots \prec w_k$ are the elements of $X$,
with $k = |X|-1$, then the sequence $w_0,\ldots,w_k$ is well-ordered.
Consequently, every vertex $x \in X$ belongs to a well-ordered
sequence, and we call \emph{index} of $x$, which we denote by
$\kappa(x)$, the smallest integer $k$ such that $x$ belongs to a
well-ordered sequence $w_0,\ldots,w_k$.  Conversely, if some vertex $x
\in V$ belongs to a well-ordered sequence $w_0, \ldots, w_k$, then an
immediate induction proves that $w_i \in X_i$ for all $i \leq k$,
whence $x \in X$.

Consider now some vertex $x \in X$.  We prove by induction on
$\kappa(x)$ that there exits a near-Dyck path from $t$ to $x$ (in the
graph $\calG$).  The result is immediate for $\kappa(x) = 0$, hence we
assume that $\kappa(x) \geq 1$. Let $w_0,\ldots,w_k$ be a well-ordered
sequence to which $x$ belongs, with $k = \kappa(x)$.  By minimality of
$\kappa(x)= k$, we know that $x = w_k$. Whenever $0 \leq i \leq k-1$,
the vertex $w_i$ must belong to $X$, and its index is at most $i$,
hence $\kappa(w_i) \leq i < \kappa(x)$. Hence, by induction hypothesis, there
exists a near-Dyck path $\pi^i$ from $t$ to $w_i$ in $\calG$:
\begin{itemize}
\item If $x \in V_\vee$, there exists a vertex $w_i$ (with $0 \leq i
  \leq k-1$) such that $(x,w_i) \in E$, hence the concatenation of
  $\pi^i$ and of the one-edge path $w_i \xrightarrow{\bullet} x$ is a
  near-Dyck path from $t$ to $x$.
\item If $x \in V_\wedge$, let $w_{1},\ldots,w_{a}$ be the elements of
  the set $\{z \in V \mid (x,z) \in E\}$ with $\theta^{-1}(w_{1}) <
  \ldots < \theta^{-1}(w_{a})$.  The concatenation of the paths
  $\pi^{a}, w_{a} \xrightarrow{w_{a}} t, \pi^{a-1}, \ldots, \pi^{1},
  w_{1} \xrightarrow{w_{1}} t$ and $\inn(x)$ is a near-Dyck path from
  $t$ to $x$ (since $\inn(x)$ is labeled by a word in $\bullet^+ \cdot
  \ol{w_1} \cdot \bullet^* \cdot \ol{w_2} \cdot \bullet^* \ldots
  \bullet^* \ol{w_a} \cdot \bullet^+$).
\end{itemize}

Conversely, let $x \in V$ be a vertex such that there exists a
near-Dyck path from $t$ to $x$ in $\calG$.  Let $\pi^x$ be a shortest
such path, and let $l(x)$ be the length of $\pi^x$.  We prove by
induction on $l(x)$ that $x \in X$.  The result is immediate for $l(x)
= 0$, hence we assume that $l(x) \geq 1$:
\begin{itemize}
\item If $x \in V_\vee$, then $\pi^x$ is the concatenation of some
  path $\rho$ and of the edge $y \xrightarrow{\bullet} x$, for some $y
  \in V$. Hence $\rho$ is a near-Dyck path of length $l(x)-1$,
  which proves that $l(y) < l(x)$, whence $y \in X$ (by induction
  hypothesis) and thus $x \in X$.
\item If $x \in V_\wedge$, then $\pi^x$ must end with $\inn(x)$. The
  label of $\inn(x)$ is a word belonging to $\bullet^+ \cdot
  \ol{w_1} \cdot \bullet^* \cdot \ol{w_2} \cdot \bullet^* \ldots
  \bullet^* \ol{w_a} \cdot \bullet^+$, where
  $\{w_{1},w_{2},\ldots,w_{a}\} = \{y \in V \mid (x,y) \in E\}$ and
  $\theta^{-1}(w_1) < \ldots < \theta^{-1}(w_a)$.
  Since $w_i \xrightarrow{w_i} t$ is the single edge labeled with
  $w_i$, we can write $\pi^x$ as a concatenation of paths $\rho^{w_a},
  w_a \xrightarrow{w_a} t, \rho^{w_{a-1}}, \ldots, \rho^{w_1}, w_1
  \xrightarrow{w_1} t, \rho^{t}, \inn(x)$, where every path $\rho^z$
  is a near-Dyck path from $t$ to $z$ (for $z \in V$).  By induction
  (all such paths $\rho^z$ have length smaller than $l(x)$), it
  follows immediately that all of $w_1,\ldots,w_k$ belong to $X$,
  whence $x \in X$. 
\end{itemize}
\end{proof}

The graph $\calG$ is \FO-definable as a function of $G$, of $t$ and of
the bijection $\theta : \{1,\ldots,n\} \mapsto V$,
and adding\slash deleting an edge $e$ in $E$ amounts to
adding\slash deleting exactly either one or two edges in $\calE_2$.
Since $\theta$ can be precomputed in \PTIME,
and due to Proposition~\ref{pro:alternating-to-near-dyck},
the alternating reachability problem is therefore bfo\textsuperscript{+}-reducible
to the near-Dyck reachability problem.

\section{Two-letter undirected Dyck reachability problem}

We proceed by proving that there exists a \bfop reduction
from the $2$-letter Dyck reachability problem (in directed graphs)
to the $2$-letter undirected reachability problem.

Let $G = (V,E,L)$ be a directed labeled graph, with
$L = \{\ell_1,\ell_2,\oell_1,\oell_2\}$, and let
$s$ and $t$ be two marked nodes of $G$.
In addition, let $\calL = \{0,1,\ol{0},\ol{1}\}$ be another set of labels.

The main difficulty, when working in an undirected graph, is that lots
of cycles are created, generating lots of stuttering in the words
labeling the paths. It is therefore hard to really control where a
path goes just by looking at its label.
In particular, the recipe used in Section~\ref{section:near-dyck-to-dyck}
to reduce the near-Dyck reachability problem to the $2$-letter Dyck reachability problem
cannot be used now, and it is not clear whether simple alternative reductions
from the near-Dyck undirected reachability problem to the $2$-letter undirected
Dyck reachability problem
exist. Below, we prove directly the \PTIME-hardness of the $2$-letter undirected Dyck
reachability. In order to do so, we rely on a rather intricate encoding, where
each part plays an important role.

We denote by $\varphi : L^\ast \mapsto \calL^\ast$ the homomorphism of monoids defined by:
\begin{xalignat*}4
\varphi(\ell_1) &= 0 \cott \BOX{\ol{0} \cott 1} \cott 1 \cott
0 \cott 0 \cott 1 \cott 1 \cott 1 \cott 1 \cott \BOX{\ol{1} \cott 0} &
\varphi(\oell_1) &= \BOX{\ol{0} \cott 1} \cott \ol{1}
\cott \ol{1} \cott \ol{1} \cott \ol{1} \cott \ol{0} \cott \ol{0} \cott
\ol{1} \cott \BOX{\ol{1} \cott 0} \cott \ol{0} \\
\varphi(\ell_2) &= 0 \cott \BOX{\ol{0} \cott 1} \cott 0 \cott
0 \cott 1 \cott 1 \cott 0 \cott 0 \cott 1 \cott \BOX{\ol{1} \cott 0} &
\varphi(\oell_2) &= \BOX{\ol{0} \cott 1} \cott \ol{1}
\cott \ol{0} \cott \ol{0} \cott \ol{1} \cott \ol{1} \cott \ol{0} \cott
\ol{0} \cott \BOX{\ol{1} \cott 0} \cott \ol{0}
\end{xalignat*}
Observe that the words $\varphi(\oell_1)$ and $\varphi(\oell_2)$
are formal inverses of the words $\varphi(\ell_1)$ and
$\varphi(\ell_2)$:
in particular, both the words $\varphi(\ell_1) \cdot \varphi(\oell_1)$ and
$\varphi(\ell_2) \cdot \varphi(\oell_2)$ are Dyck words.

In gray boxes
are locks: \emph{along a Dyck path}, once a lock has been
traveled through, we cannot go back earlier in the encoding,
since this would create a factor $1 \cott \ol{0}$ or $0 \cott \ol{1}$,
which is not a factor of any Dyck word.
We therefore say that a path is \emph{doomed} if it crosses a
lock backwards, thereby having a factor $1 \cott \ol{0}$ or $0 \cott \ol{1}$.
By preventing Dyck paths from having doomed subpaths, locks will allow
us to recover partially the directed character of $G$.

Finally, for every word $w \in \calL^\ast$, we denote by $w_i$ the $i$\textsuperscript{th} letter of $w$.
Let $\calG = (\calV,\calE,\calL)$ be the undirected labeled graph defined by:
\begin{itemize}
 \item $\calV = V \cup V \times L \times V \times \{1,\ldots,11\}$;
 \item $\calE = \calE_{\init} \cup \calE_{\mmid} \cup \calE_{\eend}$, where
\end{itemize}
\begin{equation*}
  \begin{split}
  \calE_{\init} =~& \{x \edge{\varphi(\lambda)_1} (x,\lambda,y,1) \mid (x,\lambda,y) \in E\} \\
  \calE_{\mmid} =~& \{(x,\lambda,y,i-1) \edge{\varphi(\lambda)_i} (x,\lambda,y,i) \mid (x,\lambda,y) \in E, 1 \leq i \leq 11\} \\
  \calE_{\eend} =~& \{(x,\lambda,y,11) \edge{\varphi(\lambda)_{12}} y \mid (x,\lambda,y) \in E\},
  \end{split}
\end{equation*}
and in which we mark the vertices $s$ and $t$.

Like in Sections~\ref{section:near-dyck-to-dyck} and~\ref{section:alternating-to-near-dyck},
we observe an equivalence between the two kinds of
Dyck reachability problems in the graphs $G$ and $\calG$,
which we prove formally in the rest of the section.

\begin{proposition}
\label{pro:dyck-to-undirected-dyck}
  There exists a Dyck path from $s$ to $t$ in $G$ if and only if there
  exists a Dyck path from $s$ to $t$ in $\calG$.
\end{proposition}

A first \emph{completeness} result towards proving Proposition~\ref{pro:dyck-to-undirected-dyck} comes quickly.

\begin{lemma}
\label{lem:psi-Dyck}
Let $\rho$ be a Dyck path from $s$ to $t$ in $G$.
There exists a Dyck path from $s$ to $t$ in $\calG$.
\end{lemma}

\begin{proof}
  First, for every pair $(u,v) \in \calV^2$
  there exists at most one undirected edge between $u$ and $v$ in $\calE$.
  Henceforth, we may omit representing labels of edges
  and of paths in $\calG$. 
  
  Now, let us denote by $\psi$ the mapping that identifies every edge $e \in E$
  with the path $\psi(e) = (x \to (x,\theta,y,1) \to \ldots \to
  (x,\theta,y,11) \to y)$ in $\calG$, where $e = (x,\theta,y)$.
  Observe that $\psi$ extends immediately to a morphism that maps
  every path in $G$ to a path in $\calG$.  The relation
  $\lambda(\psi(e)) = \varphi(\lambda(e))$ holds for all edges $e \in
  E$, and therefore extends to all paths in $G$.  Hence, a path $\rho$
  in $G$ is Dyck if and only if the path $\psi(\rho)$ in $\calG$ is
  Dyck. \qed
\end{proof}

However, unlike in Section~\ref{section:near-dyck-to-dyck}, there may
exist Dyck paths in $\calG$ that are not of the form $\psi(\rho)$, as
shown by the examples of the two Dyck cycles $\gamma_1$ and $\gamma_2$
displayed in Fig.~\ref{fig:dyck-graph}.
Consequently, we cannot use directly the morphism $\psi$ to associate
every Dyck path in $\calG$ with a Dyck path in $G$.

\begin{figure}[ht]
\begin{center}
\begin{tikzpicture}[scale=0.68]
\SetGraphUnit{2}
\tikzset{EdgeStyle/.append style = {>=stealth}}

\Vertex[x=0,y=0,L={$s_1$}]{1}
\EA[L={$s_2$}](1){2}
\node[anchor=south] at (1,2.7) {Graph $G$};

\Edge[label={$\ell_1$},labelstyle={below},style={bend right=45}](1)(2)
\Edge[label={$\oell_1$},labelstyle={above},style={bend right=45}](2)(1)
\Edge[style={bend right=45,->}](1)(2)
\Edge[style={bend right=45,->}](2)(1)

\SetGraphUnit{1}

\Vertex[x=4.2,y=0,L={$s_1$}]{1}
\Vertex[x=16.8,y=0,L={$s_2$}]{2}
\node[anchor=south] at (10.5,2.7) {Graph $\calG$};

\vs

\Vertex[x=5.5,y=-1.3,L={~}]{11}
\EA[L={~}](11){12}
\EA[L={~}](12){13}
\EA[L={~}](13){14}
\EA[L={~}](14){15}
\EA[L={~}](15){16}
\EA[L={~}](16){17}
\EA[L={~}](17){18}
\EA[L={~}](18){19}
\EA[L={~}](19){110}
\EA[L={~}](110){111}

\Edge[label={$0$},labelstyle={below},style={bend right=45}](1)(11)
\Edge[label={$\ol{0}$},labelstyle={below}](11)(12)
\Edge[label={${\color{white}\ol{0}}1{\color{white}\ol{0}}$},labelstyle={below}](12)(13)
\Edge[label={${\color{white}\ol{0}}1{\color{white}\ol{0}}$},labelstyle={below}](13)(14)
\Edge[label={${\color{white}\ol{0}}0{\color{white}\ol{0}}$},labelstyle={below}](14)(15)
\Edge[label={${\color{white}\ol{0}}0{\color{white}\ol{0}}$},labelstyle={below}](15)(16)
\Edge[label={${\color{white}\ol{0}}1{\color{white}\ol{0}}$},labelstyle={below}](16)(17)
\Edge[label={${\color{white}\ol{0}}1{\color{white}\ol{0}}$},labelstyle={below}](17)(18)
\Edge[label={${\color{white}\ol{0}}1{\color{white}\ol{0}}$},labelstyle={below}](18)(19)
\Edge[label={${\color{white}\ol{0}}1{\color{white}\ol{0}}$},labelstyle={below}](19)(110)
\Edge[label={$\ol{1}$},labelstyle={below}](110)(111)
\Edge[label={$0$},labelstyle={below},style={bend right=45}](111)(2)

\Edge[style={bend right=45}](1)(11)
\Edge[style={bend right=45}](111)(2)

\Vertex[x=5.5,y=-1.3,L={~}]{11}
\EA[L={~}](11){12}
\EA[L={~}](12){13}
\EA[L={~}](13){14}
\EA[L={~}](14){15}
\EA[L={~}](15){16}
\EA[L={~}](16){17}
\EA[L={~}](17){18}
\EA[L={~}](18){19}
\EA[L={~}](19){110}
\EA[L={~}](110){111}

\Vertex[x=5.5,y=1.3,L={~}]{11}
\EA[L={~}](11){12}
\EA[L={~}](12){13}
\EA[L={~}](13){14}
\EA[L={~}](14){15}
\EA[L={~}](15){16}
\EA[L={~}](16){17}
\EA[L={~}](17){18}
\EA[L={~}](18){19}
\EA[L={~}](19){110}
\EA[L={~}](110){111}

\Edge[label={$\ol{0}$},labelstyle={above},style={bend left=45}](1)(11)
\Edge[label={$0$},labelstyle={above}](11)(12)
\Edge[label={$\ol{1}$},labelstyle={above}](12)(13)
\Edge[label={$\ol{1}$},labelstyle={above}](13)(14)
\Edge[label={$\ol{0}$},labelstyle={above}](14)(15)
\Edge[label={$\ol{0}$},labelstyle={above}](15)(16)
\Edge[label={$\ol{1}$},labelstyle={above}](16)(17)
\Edge[label={$\ol{1}$},labelstyle={above}](17)(18)
\Edge[label={$\ol{1}$},labelstyle={above}](18)(19)
\Edge[label={$\ol{1}$},labelstyle={above}](19)(110)
\Edge[label={$1$},labelstyle={above}](110)(111)
\Edge[label={$\ol{0}$},labelstyle={above},style={bend left=45}](111)(2)

\Edge[style={bend left=45}](1)(11)
\Edge[style={bend left=45}](111)(2)

\Vertex[x=5.5,y=1.3,L={~}]{11}
\EA[L={~}](11){12}
\EA[L={~}](12){13}
\EA[L={~}](13){14}
\EA[L={~}](14){15}
\EA[L={~}](15){16}
\EA[L={~}](16){17}
\EA[L={~}](17){18}
\EA[L={~}](18){19}
\EA[L={~}](19){110}
\EA[L={~}](110){111}

\draw[draw=red,very thick,->,>=stealth] (4.85,-0.05) -- (4.85,-0.3) arc (180:270:0.65) arc (-90:90:0.15) arc (270:180:0.35)
-- (5.15,0.3) arc (180:90:0.35) arc (-90:90:0.15) arc (90:180:0.65) -- (4.85,0.05);

\draw[draw=black!50!green,very thick,->,>=stealth] (3.55,-0.05) -- (3.55,-1) arc (180:270:1) -- (11.5,-2) arc (90:-90:0.15) -- (9.5,-2.3) arc (90:270:0.15)
-- (13.35,-2.6) -- (15.85,-2.6) arc (-90:0:1.6) -- (17.45,1) arc (0:90:1.6) -- (7.5,2.6) arc (90:270:0.15)
-- (9.5,2.3) arc (90:-90:0.15) -- (4.55,2) arc (90:180:1) -- (3.55,0.05);

\node[anchor=west] at (5.15,0) {\color{red}\textbf{Dyck cycle $\gamma_1$}};
\node[anchor=north west] at (5.15,-2) {\color{black!50!green}\textbf{Dyck cycle $\gamma_2$}};
\end{tikzpicture}
\end{center}
\vspace{-2mm}
\caption{Graphs $G$ and $\calG$, and Dyck cycle in $\calG$}
\label{fig:dyck-graph}
\end{figure}

We overcome this problem as follows. Let $\calQ$ be the set of all
factors of all Dyck words with letters in $\calL$ (called
  \emph{approximate Dyck words}), and let $\calP$ the set of all
paths $\pi$ in $\calG$ such that $\lambda(\pi) \in \calQ$ (called
  \emph{approximate Dyck paths}).  Moreover, for every set $S$ of
paths, we denote by $\lambda(S)$ the set of labels of paths in $S$,
i.e. $\lambda(S) = \{\lambda(\pi) \mid \pi \in S\}$.  It comes at once
that $\lambda(\calP) \subseteq \calQ$, that $\calQ$ is factor
closed, and that none of the words $1 \cdot \ol{0}$ nor $0 \cdot
\ol{1}$ belongs to $\calQ$.

We further say that a path in $\calG$ is \emph{nominal} if
its source and sink belong to $V$, while its intermediate vertices
belong to $\calV \setminus V$.
For all edges $(x,\lambda,y) \in E$, we denote by
$\calP_{x,\lambda,y}$ the set of nominal paths $\pi \in \calP$ such
that $\pi$ has source $x$, sink $y$, and such that its internal
vertices are exactly the elements of the set $\{(x,\lambda,y,i) \mid 1
\leq i \leq 11\}$.  For all vertices $x \in V$, we also denote by
$\calP_x$ the set of nominal paths $\pi \in \calP$ such that $\pi$ has
source and sink $x$, and whose edges are all labeled with $0$ or
$\ol{0}$.  These two classes of paths capture the entire family of
nominal paths that belong to $\calP$, as shown by the following
result.

\begin{lemma}
\label{lem:nominal-in-P}
Let $\pi \in \calP$ be a nominal path in $\calG$. Either there exists an edge $(x,\lambda,y) \in E$ such that
$\pi \in \calP_{x,\lambda,y}$ or there exists a vertex $x \in V$ such that $\pi \in \calP_x$.
Moreover, the sets $\calP_{x,\lambda,y}$ and $\calP_x$ are pairwise disjoint.
\end{lemma}

\begin{proof}
We first assume that some edge $e$ in $\pi$ is labeled by $1$ or $\ol{1}$.
By construction, there exists a unique edge $(x,\lambda,y) \in E$ and a unique pair of integers
$i, j \in \{1,\ldots,11\}$ such that $e = (x,\lambda,y,i) \to (x,\lambda,y,j)$, with $j = i \pm 1$.
Since $\pi$ is nominal, its internal vertices belong to the set $\{(x,\lambda,y,i) \mid 1 \leq i \leq 11\}$,
and its source and sink belong to $\{x,y\}$.
Then, since $\pi$ belongs to $\calP$, it does not contain doomed paths,
hence its source must be $x$ and its sink must be $y$.

Then, assume that no edge in $\pi$ is labeled by $1$ or $\ol{1}$.
Since $\pi$ is nominal, there exists a unique edge $(x,\lambda,y) \in E$ such that
the internal vertices of $\pi$ belong to the set $\{(x,\lambda,y,i) \mid 1 \leq i \leq 11\}$,
and its source and sink belong to $\{x,y\}$.
If $x$ is the source of $\pi$, then $\pi$ can never reach the vertex $(x,\lambda,y,3)$, hence $x$ is the sink of $\pi$;
if $y$ is the source of $\pi$, then $\pi$ can never reach the vertex $(x,\lambda,y,9)$, hence $y$ is the sink of $\pi$.

Observing that every path in every set $\calP_{x,\lambda,y}$ contains
an edge labeled by $1$ or $\ol{1}$ completes the proof. \qed
\end{proof}

Going further, we associate with every path $\rho =
(v_1,\lambda_1,w_1) \cdot \ldots \cdot (v_k,\lambda_k,w_k)$ in $G$ the
set $\ol{\calP}_\rho$ of paths in $\calG$ defined by:
\[\ol{\calP}_\rho = \calP_{v_1}^\ast \cdot \calP_{v_1,\lambda_1,w_1} \cdot \calP_{v_2}^\ast \cdot \calP_{v_2,\lambda_2,w_2} \cdot \ldots
\cdot \calP_{v_k}^\ast \cdot \calP_{v_k,\lambda_k,w_k} \cdot \calP_{w_k}^\ast.\]
Observe that, unlike the sets $\calP_x$ and $\calP_{x,\lambda,y}$,
the sets $\ol{\calP}_\rho$ may contain paths that are not nominal and/or not approximate Dyck paths.

Conversely, however, it comes immediately that every Dyck path $\pi$ in $\calG$
belongs to one unique set $\ol{\calP}_\rho$, where $\rho$ is the \emph{nominal ancestor}
of $\pi$ defined below.

\begin{definition}
\label{dfn:original-path}
Let $\pi$ be a Dyck path in $\calG$ from $s$ to $t$. There exists a
unique sequence of vertices $v_0,\ldots,v_k$, a unique partial
function $f_\pi : \{1,\ldots,k\} \mapsto L$, whose domain is
denoted by $\dom(f_\pi)$, and a unique sequence of
nominal paths $\pi_1,\ldots,\pi_k$ such that:
\begin{itemize}
 \item $v_0 = s$ and $v_k = t$;
 \item for all $i \in \dom(f_\pi)$, the edge $(v_{i-1},f_\pi(i),v_i)$ belongs to $E$, and $\pi_i \in \calP_{v_{i-1},f_\pi(i),v_i}$;
 \item for all $i \in \{1,\ldots,k\} \setminus \dom(f_\pi)$, we have $v_{i-1} = v_i$, and $\pi_i \in \calP_{v_i}$;
 \item $\pi = \pi_1 \cdot \ldots \cdot \pi_k$.
\end{itemize}

We call \emph{nominal vertex sequence} of $\pi$ sequence $v_0,\ldots,v_k$,
\emph{nominal label mapping} of $\pi$ the mapping $f_\pi$,
\emph{nominal decomposition} of $\pi$ the sequence $\pi_1,\ldots,\pi_k$,
and \emph{nominal ancestor} of $\pi$ the path $(v_{i-1},f_\pi(i),v_i)_{i \in \dom(f_\pi)}$.
\end{definition}

Associating every Dyck path in $\calG$ to a unique path in $G$
is a first step towards proving the soundness of the construction.
Further steps depend on the following, additional properties of the encoding.

Every Dyck path
traveling through the word $\varphi(\ell_1)$, may go back and forth
arbitrarily, except at locks, which it may cross only once.  Consider
such a possible journey through the word $\varphi(\ell_1)$, and
observe the word $w$ obtained during that journey.  This word is made
of blocks that consist, alternatively, of letters $0$ and $\ol{0}$,
and of letters $1$ and $\ol{1}$. Such a word, \emph{even if we first
reduce it} (by deleting recursively the words $0 \cdot \ol{0}$ and $1 \cdot \ol{1}$),
will always satisfy the following properties:
\begin{enumerate}
 \item there exists at least four such (non-empty) blocks;
 \item the last block consists of letters $0$ only;
 \item the last two blocks are of odd length, and every other block is of even length.
\end{enumerate}

To illustrate the above analysis, consider the direct journey
through $\varphi(\ell_1)$, where we have identified the blocks:
$(0 \cdot \ol{0}) \cdot (1 \cdot 1) \cdot (0 \cdot 0) \cdot (1 \cdot 1 \cdot
1 \cdot 1 \cdot \ol{1}) \cdot (0)$.
If that word is reduced, there remain four non-empty blocks:
$(1 \cdot 1) \cdot (0 \cdot 0) \cdot (1 \cdot 1 \cdot 1) \cdot (0)$.

Another example is the journey first followed by the path $\gamma_2$ (see Fig.~\ref{fig:dyck-graph})
from the vertex $s_1$ to the vertex $s_2$, and which gives us more blocks:
$(0 \cdot \ol{0}) \cdot (1 \cdot 1) \cdot (0 \cdot 0) \cdot (1 \cdot 1) \cdot
(0 \cdot 0) \cdot (1 \cdot 1 \cdot 1 \cdot 1 \cdot \ol{1}) \cdot (0)$.
If that word is reduced, there remain six non-empty blocks:
$(1 \cdot 1) \cdot (0 \cdot 0) \cdot (1 \cdot 1) \cdot
(0 \cdot 0) \cdot (1 \cdot 1 \cdot 1) \cdot (0)$.
One checks easily on these examples that all the words obtained satisfy the properties 1--3.

Properties 1--2 hold for the word $\varphi(\ell_2)$, while
property 3 should be replaced by:
\begin{enumerate}
 \item[3'.] the first block of letters $1$ and $\ol{1}$
 and the last block of letters $0$
 are of odd length,
 and every other block is of even length.
\end{enumerate}
This distinction between encodings of the two letters will allow
identifying a path that encodes $\ell_1$ or $\ell_2$, even when
there is backtracking between the two locks.

Now, we denote by $\calQ_{\init}$ the set of all prefixes of all Dyck words with letters in $\calL$.
Observe that, for all words $\rho_1, \rho_2 \in \calL^\ast$,
the three words $\rho_1 \cdot \rho_2$, $\rho_1 \cdot 0 \cdot \ol{0} \cdot \rho_2$
and $\rho_1 \cdot 1 \cdot \ol{1} \cdot \rho_2$ are
either all Dyck or all non-Dyck.

Then, for every word $w \in \calL^\ast$, we call \emph{reduced word} of $w$
the word $\red(w)$ obtained from $w$ by deleting recursively the $2$-letter words
$0 \cdot \ol{0}$ or $1 \cdot \ol{1}$.
Alternatively, if we consider $w$ as an element of the free
group generated by $\ell_1$ and $\ell_2$ (with inverses $\oell_1$ and $\oell_2$),
then $\red(w)$ is the reduced word representing $w$.
We just proved that $w$ is Dyck if and only if $\red(w)$ is Dyck.
Moreover, it comes immediately that $\calQ_{\init}$ is in fact the set
of all words $w$ such that $\red(w)$ has only letters $0$ and $1$.

The above remarks and notions lead to the following results,
whose proofs are then technical yet simple, and therefore omitted here.

\begin{lemma}
\label{lem:q-init-dyck-init}
  Let $\rho$ be a path in $G$.
  If $\rho$ is not an approximate Dyck path, then
  $\lambda(\ol{\calP}_\rho) \cap \calQ = \emptyset$,
  and if $\lambda(\rho)$ is not a prefix of a Dyck word,
  then $\lambda(\ol{\calP}_\rho) \cap \calQ_\init = \emptyset$.
\end{lemma}

\begin{lemma}
\label{lem:counter-paths-core}
Let $\pi$ be a Dyck path from $s$ to $t$ in $\calG$,
let $\rho$ be the nominal ancestor of $\pi$,
and let $\lambda(\rho) \in L^\ast$ be the label of $\rho$.
In addition, let $\mu : L^\ast \mapsto \bbZ$ be the morphism of monoids defined by
$\mu(\ell_1) = \mu(\ell_2) = 1$ and $\mu(\oell_1) = \mu(\oell_2) = -1$.
Then, we have $\mu(\lambda(\rho)) = 0$.
\end{lemma}

A consequence of Lemmas~\ref{lem:q-init-dyck-init} and~\ref{lem:counter-paths-core} is
the \emph{correctness} of the construction, which is therefore valid.

\ifarxiv\begin{proof}[Proof of Proposition~\ref{pro:dyck-to-undirected-dyck}]%
\else\begin{proof}[Proposition~\ref{pro:dyck-to-undirected-dyck}]\fi
First, if there exists a Dyck path from $s$ to $t$ in $G$,
then Lemma~\ref{lem:psi-Dyck} already states that there also exists a
Dyck path from $s$ to $t$ in $\calG$. Hence, we look at the converse implication.

Let $\pi$ be a Dyck path from $s$ to $t$ in $\calG$, and
let $\rho$ be the nominal ancestor of $\pi$. Let $\lambda(\rho)$
be the label of $\rho$ and let $\Lambda$ be the reduction of $\lambda(\rho)$.
Lemma~\ref{lem:q-init-dyck-init} proves that $\lambda(\rho)$ is
a prefix of a Dyck word, hence that $\Lambda$ has only letters $\ell_1$ and $\ell_2$.
Since Lemma~\ref{lem:counter-paths-core} also proves that $\mu(\lambda(\rho)) = \mu(\Lambda) = 0$,
it follows that $\Lambda$ is the empty word, i.e. that $\lambda(\rho)$ is a Dyck word. \qed
\end{proof}

We complete the proof of Theorem~\ref{thm:main} as follows.
Observe that the graph $\calG$ is \FO-definable as a function of $G$.
Furthermore, adding\slash deleting an edge in $E$ amounts to
adding\slash deleting exactly twelve edges in $\calE$.
Due to Proposition~\ref{pro:dyck-to-undirected-dyck}, the $2$-letter Dyck reachability problem is therefore \bfo-reducible
to the $2$-letter undirected Dyck reachability problem.

On the other hand, as a restriction of the $n$-letter Dyck reachability problem,
the $n$-letter undirected Dyck reachability problem is clearly in \PTIME, which completes the proof of Theorem~\ref{thm:main}.

\bibliographystyle{plain}

\pagebreak

\appendix

\centerline{\noindent{\LARGE\bfseries --- Appendix: Proving Proposition~\ref{pro:dyck-to-undirected-dyck} ---}}

\bigskip\bigskip

\label{app:lem-1-0-omega}

The following proofs of Lemmas~\ref{lem:q-init-dyck-init} and~\ref{lem:counter-paths-core}
are mainly based on the notions of nominal ancestors and of reductions, as well as
on the sets $\calQ$ (of approximate Dyck words) and
$\calQ_{\init}$ (of all prefixes of all Dyck words).
We refer the reader to above definitions and properties of these concepts.

Aiming to avoid using heavy notations, we identify below regular expressions
with the languages they represents and,
for all sets of words $A$ and $B$, we write $A \subred B$ as a placeholder for
$\{\red(w) \mid w \in A\} \subseteq B$.

Furthermore, recall that no word in $\calQ$ contains the sub-words $1 \cdot \ol{0}$ nor $0 \cdot \ol{1}$,
and that no word whose first letter is $\ol{0}$ or $\ol{1}$ can belong to $\calQ_{\init}$.
Hence, in a reduced word $w \in \calQ$, no letter $0$ or $1$ can be followed by a letter $\ol{0}$ or $\ol{1}$ and,
if $w \in \calQ_{\init}$, then $w$ contains no letter $\ol{0}$ or $\ol{1}$. 

In addition, consider the following regular expressions:
\begin{xalignat*}6
& \omega_+ = (0 \cdot 0 + 1 \cdot 1)^\ast &
& \varpi_+ = (\omega_+ \cdot 1 \cdot \omega_+ \cdot 1)^\ast \cdot \omega_+ \\
& \omega_- = (\ol{0} \cdot \ol{0} + \ol{1} \cdot \ol{1})^\ast &
& \varpi_- = (\omega_- \cdot \ol{1} \cdot \omega_- \cdot \ol{1})^\ast \cdot \omega_- \\
& \omega = (\omega_+ + \omega_- + \ol{0} \cdot 0)^\ast ~~~~ &
& \varpi = (\omega \cdot 1 \cdot \omega \cdot \ol{1})^\ast \cdot \omega
\end{xalignat*}

The intuition behind these expressions is as follows.
Roughly speaking, we prove below that:
\begin{itemize}
 \item every circuit in $\calG$ whose edges are all labeled by $0$ or $\ol{0}$ has its label in $\omega$, and therefore in $\varpi$;
 \item if $\rho$ is an approximate Dyck path in $G$, then every approximate
   Dyck path in $\calP_\rho$ has its label in $\varpi$.
\end{itemize}
These statements are proved by induction, hence $\varpi$ was found as a fixed point
for some closure properties similar to those of Dyck paths.
Then, $\omega_+$, $\omega_-$, $\varpi_+$ and $\varpi_-$ are just specializations of $\omega$ and $\varpi$
where only labels $0$ or $1$ (respectively, $\ol{0}$ or $\ol{1}$) are allowed.

These latter regular expressions will allow us to express the following auxiliary results,
which will lead us to Lemma~\ref{lem:1-0-omega-3}.

\begin{lemma}
\label{lem:1-0-omega-2}
The following relations hold:
\[(1 \cdot 0 \cdot \varpi) \cap \calQ \subred 1 \cdot 0 \cdot \varpi_+ ~~~~~~~~
(\varpi \cdot \ol{0} \cdot \ol{1}) \cap \calQ \subred \varpi_- \cdot \ol{0} \cdot \ol{1} 
\]
\end{lemma}

\begin{proof}
Consider the regular expression $\alpha =
1 \cdot (0 + \ol{0}) \cdot (0 \cdot 0 + \ol{0} \cdot 0 + 0 \cdot \ol{0} + \ol{0} \cdot \ol{0} + 1 + \ol{1})^\ast$.
It comes at once that $1 \cdot 0 \cdot \varpi \subseteq \alpha$.
Furthermore, for all words $\rho_1, \rho_2$ such that either
$\rho_1 \cdot 0 \cdot \ol{0} \cdot \rho_2$ or $\rho_1 \cdot 1 \cdot \ol{1} \cdot \rho_2$ belongs to $\alpha$,
the word $\rho_1 \cdot \rho_2$ belongs to $\alpha$ too, which means that $\alpha$ is closed under reduction.

Let $\rho$ be a word in $\varpi$ such that $1 \cdot 0 \cdot \rho \in \calQ$,
and let $\rho' = \red(1 \cdot 0 \cdot \rho)$.
Since $\rho$ belongs to $\alpha$, so does $\rho'$,
whose leftmost letter must then be $1$.
We remarked above that, since $\rho' \in \calQ$, no letter $0$ or $1$ in $\rho'$ can be followed by a letter $\ol{0}$ or $\ol{1}$.
It follows that $\rho'$ does not contain any letter $\ol{0}$ or $\ol{1}$, i.e. that $\rho' = 1 \cdot 0 \cdot \rho''$
for some reduced word $\rho''$. Then, $\rho''$ must be the reduction of $\rho$, and cannot contain letters 
$\ol{0}$ or $\ol{1}$, hence $\rho'' \in \varpi_+$, i.e. $\rho' \in 1 \cdot 0 \cdot \varpi_+$.
We prove similarly the relation $(\varpi \cdot \ol{0} \cdot \ol{1})
\cap \calQ \subred \varpi_- \cdot \ol{0} \cdot \ol{1}$. \qed
\end{proof}

This first result allows us to restate, in a more precise manner,
the properties 1--3 and 3' stated in the core of the paper
(as well as their variants for the inverse words $\varphi(\oell_1)$ and $\varphi(\oell_2)$).

\begin{lemma}
\label{lem:1-0-omega}
For all vertices $x, y \in V$, we have $\lambda(\calP_x) \subred \varpi$ and
\begin{xalignat*}4
  \lambda(\calP_{x,\ell_1,y}) & \subred \varpi \cott 1 \cott 1 \cott 0 \cott 0 \cott \omega_+ \cott 1 \cott 0 &
  \lambda(\calP_{x,\ell_2,y}) & \subred \varpi \cott 1 \cott \omega_+ \cott 0 \cott 0 \cott 1 \cott 1 \cott \omega_+ \cott 0 \\
  \lambda(\calP_{x,\oell_1,y}) & \subred \ol{0} \cott \ol{1} \cott \omega_- \cott \ol{0} \cott \ol{0} \cott \ol{1} \cott \ol{1} \cott \varpi &
  \lambda(\calP_{x,\oell_2,y}) & \subred \ol{0} \cott \omega_- \cott \ol{1} \cott \ol{1} \cott \ol{0} \cott \ol{0} \cott \omega_- \cott \ol{1} \cott \varpi
\end{xalignat*}
\end{lemma}

\begin{proof}
First, let $\pi$ be a path in $\calP_x$. By construction, the graph
$\calG$ is bipartite, hence $\pi$ is of even length.  Its edges are
labeled by $0$ and $\ol{0}$ only, hence the word $\lambda(\pi)$,
once reduced, belongs to $\omega$, and therefore to $\varpi$.

Then, let $\pi$ be a path in $\calP_{x,\ell_1,y}$. By construction, $\pi$ belongs to $\calP$,
hence its label $\lambda(\pi)$ belongs to $\calQ$, and $\pi$ cannot contain any doomed path.
Consequently, we can decompose $\pi$ as a concatenation of paths $\pi_1 \cdot \pi_2 \cdot \pi_3 \cdot \pi_4$, where:
\begin{itemize}
 \item $\pi_1$ goes from $x$ to $(x,\ell_1,y,2)$, with internal vertices in $\{(x,\ell_1,y,i) \mid 1 \leq i \leq 2\}$;
 \item $\pi_2$ goes from $(x,\ell_1,y,2)$ to $(x,\ell_1,y,6)$, with internal vertices in $\{(x,\ell_1,y,i) \mid 2 \leq i \leq 6\}$;
 \item $\pi_3$ goes from $(x,\ell_1,y,6)$ to $(x,\ell_1,y,11)$, with internal vertices in $\{(x,\ell_1,y,i) \mid 2 \leq i \leq 11\}$;
 \item $\pi_4$ is the one-letter path $(x,\ell_1,y,11) \to
   y$ (by definition of a nominal path, we stop at the first occurrence of $y$).
\end{itemize}
It comes at once that $\lambda(\pi_1) \in \varpi$, that $\lambda(\pi_2) \in 1 \cdot 1 \cdot \omega_+ \cdot 0 \cdot 0$ and that
$\lambda(\pi_4) = 0$. 
The word $\lambda(\pi_3)$ belongs a
priori to $\omega \cdot \ol{1}$, but we show that when reduced, it
actually contains no occurrence of letters $\ol{0}$ and
$\ol{1}$. The case of $\ol{0}$ is obvious.

Assume now that there is some
$\ol{1}$ in the reduced version of $\lambda(\pi_3)$.
The leftmost letter $\ol{1}$ cannot be preceded by a $1$ (since $\lambda(\pi_3)$ is reduced)
nor by a $0$ (since $0 \cdot \ol{1} \notin \calQ$). Hence, it must be the first letter of $\lambda(\pi_3)$.
However, since $0$ is the last letter of $\lambda(\pi_2)$, the concatenated word
$\lambda(\pi) = \lambda(\pi_1) \cdot \lambda(\pi_2) \cdot \lambda(\pi_3) \cdot \lambda(\pi_4)$
cannot belong to $\calQ$ itself, which is a contradiction. 
We conclude that the reduced version of
$\lambda(\pi_3)$ belongs to $\omega_+ \cdot 1$.
Hence, we have proved that 
\[\lambda(\calP_{x,\ell_1,y}) \subred \varpi \cdot 1 \cdot 1 \cdot \omega_+ \cdot 0 \cdot 0 \cdot \omega_+ \cdot 1 \cdot 0
= \varpi \cdot 1 \cdot 1 \cdot 0 \cdot 0 \cdot \omega_+ \cdot 1 \cdot 0.\]

The case of the language $\lambda(\calP_{x,\oell_1,y}) \cap \calQ$ is treated in the exact same manner,
and the two remaining cases are very similar as well. \qed 
\end{proof}

Using the two previous technical lemmas, we are able to show that
the set $\calQ$ of \emph{approximate Dyck words} allows to
discriminate between legal (for Dyck) and illegal consecutions of
nominal paths. Similarly, prefixes $\calQ_\init$ eliminate illegal
beginnings of Dyck paths. This is stated as follows.

\begin{lemma}
\label{lem:1-0-omega-3}
For all vertices $x, x', y, y' \in V$, we have
\begin{xalignat*}4
  & (\lambda(\calP_{x,\ell_1,y}) \cdot \varpi \cdot \lambda(\calP_{x',\oell_1,y'})) \cap \calQ \subred \varpi &
  & ~~~(\lambda(\calP_{x,\ell_2,y}) \cdot \varpi \cdot \lambda(\calP_{x',\oell_2,y'})) \cap \calQ \subred \varpi \\
  & (\lambda(\calP_{x,\ell_1,y}) \cdot \varpi \cdot \lambda(\calP_{x',\oell_2,y'})) \cap \calQ = \emptyset &
  & ~~~(\lambda(\calP_{x,\ell_2,y}) \cdot \varpi \cdot \lambda(\calP_{x',\oell_1,y'})) \cap \calQ = \emptyset \\
  & (\varpi \cdot \lambda(\calP_{x,\oell_1,y})) \cap \calQ_\init = \emptyset &
  & ~~~(\varpi \cdot \lambda(\calP_{x,\oell_2,y})) \cap \calQ_\init = \emptyset
\end{xalignat*}
\end{lemma}

\begin{proof}
First, let $\rho_1$, $\rho_2$ and $\rho_3$ be reduced words in $\lambda(\calP_{x,\ell_1,y})$, $\varpi$ and $\lambda(\calP_{x',\oell_1,y'})$
such that $\rho_1 \cdot \rho_2 \cdot \rho_3$ belongs to $\calQ$.
Lemma~\ref{lem:1-0-omega} proves that $\rho_1$ ends with the suffix $1 \cdot 0$, and that
$\rho_3$ begins with the prefix $\ol{0} \cdot \ol{1}$.

Hence, the word $\rho'_2 = 1 \cdot 0 \cdot \rho_2 \cdot \ol{0} \cdot \ol{1}$
belongs to $\calQ$, and therefore Lemma~\ref{lem:1-0-omega-2} proves that $\rho'_2$ belongs to both
$1 \cdot 0 \cdot \varpi_+ \cdot \ol{0} \cdot \ol{1}$ and $1 \cdot 0 \cdot \varpi_- \cdot \ol{0} \cdot \ol{1}$.
Since $(1 \cdot 0 \cdot \varpi_+ \cdot \ol{0} \cdot \ol{1}) \cap (1 \cdot 0 \cdot \varpi_- \cdot \ol{0} \cdot \ol{1}) \subred 1 \cdot 0 \cdot \ol{0} \cdot \ol{1} \subred \varepsilon$,
it follows that $\rho_1 \cdot \rho_2 \cdot \rho_3$, once reduced, belongs to the set
$\varpi \cdot 1 \cdot 1 \cdot 0 \cdot 0 \cdot \omega_+ \cdot \varepsilon
\cdot \omega_- \cdot \ol{0} \cdot \ol{0} \cdot \ol{1} \cdot \ol{1} \cdot \varpi$, i.e. to $\varpi$ itself.

We prove the relation $(\lambda(\calP_{x,\ell_2,y}) \cdot \varpi \cdot \lambda(\calP_{x',\oell_2,y'})) \cap \calQ \subred \varpi$
in a similar way.
Let $\rho_1$, $\rho_2$ and $\rho_3$ be reduced words in $\lambda(\calP_{x,\ell_2,y})$, $\varpi$ and $\lambda(\calP_{x',\oell_2,y'})$
such that $\rho_1 \cdot \rho_2 \cdot \rho_3$ belongs to $\calQ$.
Lemma~\ref{lem:1-0-omega} proves that $\rho_1$ belongs to $\varpi \cdot 1 \cdot \omega_+ \cdot 0 \cdot 0 \cdot \omega_+ \cdot 1 \cdot 1 \cdot 0^{2a+1}$
and that $\rho_3$ belongs to $\ol{0}^{2b+1} \cdot \ol{1} \cdot \ol{1} \cdot \omega_- \cdot \ol{0} \cdot \ol{0} \cdot \omega_- \cdot \ol{1} \cdot \varpi$
for some integers $a, b \geq 0$.

Once again, Lemma~\ref{lem:1-0-omega-2} proves that the word $\rho'_2 = 1 \cdot 0^{2a+1} \cdot \rho_2 \cdot \ol{0}^{2b+1} \cdot \ol{1}$
is reducible to the empty word,
which shows that $\rho_1 \cdot \rho_2 \cdot \rho_3$, once reduced, belongs to
$\varpi \cdot 1 \cdot \omega_+ \cdot 0 \cdot 0 \cdot \omega_+ \cdot 1 \cdot \varepsilon \cdot
\ol{1} \cdot \omega_- \cdot \ol{0} \cdot \ol{0} \cdot \omega_- \cdot \ol{1} \cdot \varpi$,
i.e. to $\varpi$ itself.

Second, assume that there exists reduced words $\rho_1$, $\rho_2$ and $\rho_3$
in $\lambda(\calP_{x,\ell_1,y})$, $\varpi$ and $\lambda(\calP_{x',\oell_2,y'})$
such that $\rho_1 \cdot \rho_2 \cdot \rho_3$ belongs to $\calQ$.
The same arguments as above show that $\rho_1 \cdot \rho_2 \cdot \rho_3$, once reduced, belongs to
$\varpi \cdot 1 \cdot 1 \cdot 0 \cdot 0 \cdot \omega_+ \cdot \varepsilon \cdot
\ol{1} \cdot \omega_- \cdot \ol{0} \cdot \ol{0} \cdot \omega_- \cdot \ol{1} \cdot \varpi$,
i.e. that there exists integers $a,b \geq 0$ such that
$\red(\rho_1 \cdot \rho_2 \cdot \rho_3)$ belongs to
$\varpi \cdot 1 \cdot 1 \cdot \omega_+ \cdot 0 \cdot 0 \cdot 1^{2a} \cdot \varepsilon \cdot
\ol{1}^{2b+1} \cdot \ol{0} \cdot \ol{0} \cdot \omega_- \cdot \omega_- \cdot \ol{1} \cdot \varpi$.

Every word of the form $0 \cdot 1^{2a} \cdot \ol{1}^{2b+1} \cdot \ol{0}$ is reducible to a word
$0 \cdot 1^{2x+1} \cdot \ol{0}$ or $0 \cdot \ol{1}^{2x+1} \cdot \ol{0}$, hence it cannot belong to $\calQ$.
This proves that $\rho_1 \cdot \rho_2 \cdot \rho_3$ cannot belong to $\calQ$ as well.

We prove in the exact same manner the equality
$(\lambda(\calP_{x,\ell_2,y}) \cdot \varpi \cdot \lambda(\calP_{x',\oell_1,y'})) \cap \calQ = \emptyset$.

Finally, the same arguments also show that
$(\varpi \cdot \lambda(\calP_{x,\oell_1,y})) \cap Q \subred \varpi_- \cdot \lambda(\calP_{x,\oell_1,y})$ and that
$(\varpi \cdot \lambda(\calP_{x,\oell_2,y})) \cap Q \subred \varpi_- \cdot \lambda(\calP_{x,\oell_2,y})$.
However, no word in $\calQ_\init$ begins with a letter $\ol{0}$ or $\ol{1}$,
while every word in $\varpi_-$ (if that word is non-empty), in $\lambda(\calP_{x,\oell_1,y})$ or
in $\lambda(\calP_{x,\oell_2,y})$ begins with such a letter.
Hence, both$(\varpi \cdot \lambda(\calP_{x,\oell_1,y})) \cap Q$ and
$(\varpi \cdot \lambda(\calP_{x,\oell_2,y})) \cap Q$ are empty. \qed
\end{proof}

This allows to state the partial correctness of the construction, as stated
in Lemma~\ref{lem:q-init-dyck-init} in the core of the paper,
and which we reproduce here.

\setcounter{lemma}{3}
\setcounter{mycount}{3}

\begin{lemma}
  Let $\rho$ be a path in $G$.
  If $\rho$ is not an approximate Dyck path, then
  $\lambda(\ol{\calP}_\rho) \cap \calQ = \emptyset$,
  and if $\lambda(\rho)$ is not a prefix of a Dyck word,
  then $\lambda(\ol{\calP}_\rho) \cap \calQ_\init = \emptyset$.
\end{lemma}

\begin{proof}
We begin with an auxiliary result: we first prove that $\lambda(\ol{\calP}_\rho) \cap \calQ \subred \varpi$ if
$\rho$ is a Dyck path, and proceed by induction on the length $|\rho|$ of the path $\rho$.
If $|\rho| = 0$, then $\lambda(\ol{\calP}_\rho) \subseteq \lambda(\calP_x)^\ast \subseteq \varpi^\ast = \varpi$,
where $x$ is the source and sink of $\rho$.
Then, if $|\rho| \geq 1$ and if $\rho$ is a concatenation of two non-empty Dyck paths $\rho_1$ and $\rho_2$, we have
\[\lambda(\ol{\calP}_\rho) \cap \calQ \subred (\lambda(\ol{\calP}_{\rho_1}) \cap \calQ) \cdot (\lambda(\ol{\calP}_{\rho_2}) \cap \calQ) \subred \varpi \cdot \varpi = \varpi.\]
Finally, if $|\rho| \geq 1$ and if $\rho$ is a concatenation of an edge
$(x_1,\theta,x_2)$, a Dyck path $\rho'$ with source $x_2$ and sink $y_1$, and an edge $(y_1,\ol{\theta},y_2)$ (with $\theta \in \{\ell_1,\ell_2\}$), then
\begin{equation*}
  \begin{split}
  \lambda(\ol{\calP}_\rho) \cap \calQ \subred ~& \lambda(\calP_{x_1})^\ast \cdot
      ((\lambda(\calP_{x_1,\theta,x_2}) \cdot (\lambda(\ol{\calP}_{\rho'}) \cap \calQ) \cdot
      \lambda(\calP_{y_1,\ol{\theta},y_2})) \cap \calQ) \cdot \lambda(\calP_{y_2})^\ast \\
  \subred ~& \varpi^\ast \cdot ((\lambda(\calP_{x_1,\theta,x_2}) \cdot \varpi \cdot
      \lambda(\calP_{y_1,\ol{\theta},y_2})) \cap \calQ) \cdot \varpi^\ast \\
  \subred ~& \varpi^\ast \cdot \varpi \cdot \varpi^\ast = \varpi.
  \end{split}
\end{equation*}

Then, we prove that $\lambda(\ol{\calP}_\rho) \cap \calQ = \emptyset$ if
$\rho$ is not an approximate Dyck path. Indeed, let $\rho'$ be a minimal sub-path of $\rho'$
that is not approximate Dyck. The word $\rho'$ cannot be empty, hence it must be of the form
$\rho' = (x_1,\theta,x_2) \cdot \rho'' \cdot (y_1,\chi,y_2)$, where $\rho''$ is a Dyck path and either
$\theta = \ell_1$ and $\chi = \oell_2$ or $\theta = \ell_2$ and $\ol\chi = \oell_1$.
It follows that
\begin{equation*}
  \begin{split}
  \lambda(\ol{\calP}_{\rho'}) \cap \calQ \subred ~& \lambda(\calP_{x_1})^\ast \cdot
      ((\lambda(\calP_{x_1,\theta,x_2}) \cdot (\lambda(\ol{\calP}_{\rho''}) \cap \calQ) \cdot
      \lambda(\calP_{y_1,\chi,y_2})) \cap \calQ) \cdot \lambda(\calP_{y_2})^\ast \\
  \subred ~& \varpi^\ast \cdot ((\lambda(\calP_{x_1,\theta,x_2}) \cdot \varpi \cdot
      \lambda(\calP_{y_1,\chi,y_2})) \cap \calQ) \cdot \varpi^\ast \\
  \subred ~& \varpi^\ast \cdot \emptyset \cdot \varpi^\ast = \emptyset,
  \end{split}
\end{equation*}
and therefore that $\lambda(\ol{\calP}_\rho) \cap \calQ = \emptyset$ as well.

Finally, if $\lambda(\rho)$ is not a prefix of a Dyck word,
then the relation $\lambda(\ol{\calP}_\rho) \cap \calQ_\init = \emptyset$
is immediate when $\rho$ itself is not an approximate Dyck path (since
$\calQ_\init \subseteq \calQ$).
Hence, we assume that $\rho$ is an approximate Dyck path.
Let $\rho_1$ be the longest prefix of $\rho$ such that
$\lambda(\rho_1)$ is a prefix of a Dyck word.
Let us write $\rho$ as a path of the form $\rho = \rho_1 \cdot (x,\theta,y) \cdot \rho_2$,
where $(x,\theta,y)$ is an edge of $G$.

Let $\Lambda$ be the reduction of $\lambda(\rho_1)$.
If $\theta \in \{\ell_1,\ell_2\}$, then the word $\lambda(\rho_1) \cdot \theta \cdot \ol{\theta}$
reduces to $\Lambda$, hence is a prefix of a Dyck word, contradicting the maximality of $\rho_1$.
Hence, we know that $\theta \in \{\oell_1,\oell_2\}$.

Now, let us assume that $\Lambda$ is not the empty word.
Since $\lambda(\rho_1)$ is a prefix of a Dyck word, $\Lambda$ must contain only letters
$\ell_1$ and $\ell_2$. Hence, without loss of generality, we assume that the rightmost
letter of $\Lambda$ is $\ell_1$.
If $\theta = \oell_1$, then the word $\lambda(\rho_1) \cdot \theta \cdot \ell_1$
reduces to $\Lambda$ too, which is impossible.
Finally, if $\theta = \oell_2$, then $\lambda(\rho_1) \cdot \theta$ reduces to
$\Lambda \cdot \oell_2$, which is not an approximate Dyck word,
contradicting the definition of $\rho$.

Hence, $\Lambda$ is the empty word, i.e. $\rho_1$ is a Dyck path. It follows that
\begin{equation*}
  \begin{split}
  \lambda(\ol{\calP}_{\rho_1 \cdot (x,\theta,y)}) \cap \calQ_\init \subred ~& (((\lambda(\ol{\calP}_{\rho_1}) \cap \calQ) \cdot
      \lambda(\calP_{x,\theta,y})) \cap \calQ_\init) \cdot \lambda(\calP_y)^\ast \\
  \subred ~& ((\varpi \cdot \lambda(\calP_{x,\theta,y})) \cap \calQ_\init) \cdot \varpi^\ast \\
  \subred ~& \emptyset \cdot \varpi^\ast = \emptyset,
  \end{split}
\end{equation*}
and therefore that $\lambda(\ol{\calP}_\rho) \cap \calQ_\init = \emptyset$ as well. \qed
\end{proof}

The above auxiliary lemmas also lead to the following,
more direct proof of Lemma~\ref{lem:counter-paths-core},
which we also reproduce here.

\begin{lemma}
Let $\pi$ be a Dyck path from $s$ to $t$ in $\calG$,
let $\rho$ be the nominal ancestor of $\pi$,
and let $\lambda(\rho) \in L^\ast$ be the label of $\rho$.
In addition, let $\mu : L^\ast \mapsto \bbZ$ be the morphism of monoids defined by
$\mu(\ell_1) = \mu(\ell_2) = 1$ and $\mu(\oell_1) = \mu(\oell_2) = -1$.
Then, we have $\mu(\lambda(\rho)) = 0$.
\end{lemma}

\begin{proof}
Let us identify the monoid $\calL^\ast$ with a sub-monoid of the free group $\bbZ \ast \bbZ$,
where the letters $0$, $\ol{0}$, $1$ and $\ol{1}$
are identified with the elements $(1,0)$, $(-1,0)$, $(0,1)$ and $(0,-1)$ respectively.
Then, let $\Theta$ be the canonical projection of the free group $\bbZ \ast \bbZ$ onto the free group $\bbZ_2 \ast \bbZ_2$.
Let $\gamma$ be the element $(0,1) \cdot (1,0)$ of $\bbZ_2 \ast \bbZ_2$.
Observe that $\gamma$ is not a torsion element of $\bbZ_2 \ast \bbZ_2$, and therefore that
the group $\{\gamma^i \mid i \in \bbZ\}$ is isomorphic with $\bbZ$.

It comes at once that $\Theta$ maps all of the languages $\omega_+$, $\omega_-$, $\omega$, $\varpi_+$, $\varpi_-$ and $\varpi$
to the singleton set $\{\gamma^0\}$. It follows that
$\Theta(\lambda(\calP_x)) = \{\gamma^0\}$, that $\Theta(\lambda(\calP_{x,\ell_1,y})) = \Theta(\lambda(\calP_{x,\ell_2,y})) = \{\gamma\}$ and that
$\Theta(\lambda(\calP_{x,\oell_1,y})) = \Theta(\lambda(\calP_{x,\oell_2,y})) = \{\gamma^{-1}\}$
for all $x,y \in V$.

Hence, let $v_0,\ldots,v_k$, $\pi_1,\ldots,\pi_k$ and $f_\pi : \{1,\ldots,k\} \mapsto L$
be the nominal vertex sequence, the nominal decomposition and the nominal label mapping of $\pi$.
Observe that, as an element of $\bbZ \ast \bbZ$,
$\lambda(\pi) = \lambda(\pi_1) \cdot \ldots \cdot \lambda(\pi_k)$ is the neutral element of the group,
whence $\gamma^0 = \Theta(\lambda(\pi_1)) \cdot \ldots \cdot \Theta(\lambda(\pi_k))$.

We showed above that $\Theta(\lambda(\pi_i)) = \gamma^0$ for all $i \in \{1,\ldots,k\} \setminus \dom(f_\pi)$
and that $\Theta(\lambda(\pi_i)) = \gamma^{\mu(f_\pi(i))}$ for all $i \in \dom(f_\pi)$,
which completes the proof. \qed
\end{proof}

\end{document}